\documentclass[english,12pt]{article}

\usepackage{amssymb}
\usepackage{amsfonts} 
\usepackage{amssymb}
\usepackage{amsmath}
\usepackage[utf8]{inputenc}
\usepackage{amsthm}
\usepackage[margin=0.5in]{geometry}
\usepackage{float}
\usepackage{enumerate}
\usepackage{bookmark}
\usepackage[toc,page]{appendix}

\usepackage{enumitem}

\usepackage{hyperref}
\hypersetup{
    colorlinks,
    citecolor =blue,
    linkcolor=blue,
    filecolor=blue,      
    urlcolor=blue,
}

\usepackage{xpatch}
\xpatchcmd{\proof}{\itshape}{\bfseries\proofnameformat}{}{}

\newcommand{\proofnameformat}{}

\usepackage{color}
\usepackage{graphicx}

\newcommand{\ra}{{\rightarrow}}

\newcommand{\R}{{\mathbb R}}

\newcommand{\C}{{\mathcal C}} 
\newcommand{\E}{{\mathbb E}}

\newcommand{\A}{\mathcal{A}}

\newcommand{\bi}{\begin{itemize}}
\newcommand{\ei}{\end{itemize}}

\DeclareMathOperator*{\argmax}{arg\,max}
\newtheorem{theorem}{Theorem}
\newtheorem{lemma}{Lemma}
\newtheorem{proposition}{Proposition}

\begin{document}

\title{Mechanism Design with News Utility\footnote{I am thankful to Drew Fudenberg, Matthew Rabin and Tomasz Strzalecki for continuous support in this project as well as to Fabian Herweg for introducing me to this topic. I thank Daniel Clark and Kevin He for numerous comments as well as Arjada Bardhi, Krishna Dasaratha, Jerry Green, Annie Liang and Eric Maskin for comments during different stages of this project. Any errors are mine.}}
\date{\vspace{-5ex}}

\author{Jetlir Duraj\footnote{duraj@g.harvard.edu}}
\maketitle

\begin{abstract}
News utility is the idea that 
the utility of an agent also depends on changes in her beliefs over consumption and money. We introduce news utility into otherwise classical static Bayesian mechanism design models. We show that a key role is played by the timeline of the mechanism, i.e. whether there are delays between the announcement stage, the participation stage, the play stage and the realization stage of a mechanism. Depending on the timing, agents with news utility can experience two additional news utility effects: a surprise effect derived from comparing to pre-mechanism beliefs, as well as a realization effect derived from comparing post-play beliefs with the actual outcome of the mechanism. 

We look at two distinct mechanism design settings reflecting the two main strands of the classical literature. In the first model, a monopolist screens an agent according to the magnitude of her loss aversion. In the second model, we consider a general multi-agent Bayesian mechanism design setting where the uncertainty of each player stems from not knowing the intrinsic types of the other agents. We give applications to auctions and public good provision which illustrate how news utility changes classical results.

For both models we characterize the optimal design of the timeline. A timeline featuring no delay between participation and play but a delay in realization is never optimal in either model. In the screening model the optimal timeline is one without delays. In auction settings, under fairly natural assumptions the optimal timeline may have delays between all three stages of the mechanism.





\end{abstract}

\section{Introduction}\label{sec:intro}

Most situations in practice to which the theory of classical static Bayesian mechanism design is applicable can be thought of as consisting of three distinct stages: first, the mechanism is announced to the agents and the agents decide whether to participate; second, the agents decide what to play in the mechanism and finally, the mechanism outcome consisting of a consumption allocation and money transfers to the designer is realized. 

Classical models of Bayesian mechanism design generally assume that the agents possess quasi-linear utility and are Expected Utility maximizers. Absent discounting issues the analysis is the same in the classical model for the cases where the above mentioned stages may happen with delay from each other.\footnote{Besides strategic use of delays by a designer, exogenously given delays due to technological constraints between stages where uncertainty persists in the agents' minds are a recurrent feature of life: goods need to be produced, information must travel, etc.} This is because the agents in the classical model are time-consistent as well as insensitive to the timing of the realization of uncertainty. This paper characterizes mechanism design for agents who violate the last two assumptions in a specific way: they experience \emph{news utility} and \emph{loss aversion}. These two features make the agents sensitive to whether uncertainty is resolved with delay and lead to time-inconsistent behavior.
We look at the case of news-utility agents who are sophisticated with regard to this time-inconsistency. These assumptions lead to considerable differences to classical analysis: besides changes in several key intuitions from classical settings, the issue of the optimal design of the timeline of the mechanism becomes salient.

More precisely, in this paper we assume agents possess quasi-linear \emph{intrinsic} utility and add the innovation that their utilities depend on changes in their beliefs over consumption and money (henceforth called \emph{news utility}). News utility is assumed separable in the good and money dimension and the comparison of new to old beliefs uses a classical gain-loss function featuring loss aversion (see e.g. \cite{kt}) -- relative to good news, utility losses from `bad news' are compounded due to loss aversion. The agents in this paper are forward-looking with respect to both intrinsic and news utility as well as sophisticated about their future behavior. We assume news utility is produced only from objective sources, i.e. there is no self-production of news utility and that each agent takes into account future news utility in expectation in any decision instance. Just as for the classical part of the utility, the belief over future consumption and money used to weigh news utility is induced by an agent's play and other random factors in the environment. 

We consider the consequences of such preferences in two different Bayesian mechanism design settings reflecting two of the main strands of the classical literature: monopolistic screening of a single agent as well as multi-agent mechanism design such as auction or public goods settings. In the screening model we assume the uncertainty facing the monopolist concerns a behavioral parameter of the agent: her loss aversion level in the good and money dimension. We also assume that the agent learns her intrinsic type only upon consumption. This simple model already matches many situations in the real world where the intrinsic value of a consumption good is discovered only upon consumption. In the multi-agent model we assume that all `behavioral' features of the preferences are common knowledge, the agents know their intrinsic type when presented with the mechanism, and an agent's uncertainty only comes from not knowing the intrinsic type of the other agents. Therefore, the informational side of the multi-agent model is a straightforward extension of the classical multi-agent Bayesian mechanism design model with quasilinear utility, whose prime examples in the literature are auction settings or provision of public goods.

In stark contrast to the classical setting, with the new preferences it matters whether there are delays between stages of a mechanism so that we distinguish three main timelines for the analysis.\footnote{Additional timelines are equivalent to the ones presented here under two  simplifying assumptions: no-discounting of utility and the the designer cannot randomize. See subsection \ref{sec:prefs} and Proposition \ref{thm:oneagenthelp} for more details. Figure \ref{fig:screening} from subsection \ref{sec:screeningsetup} and Figure \ref{fig:tl} from subsection \ref{sec:multiagentsetup} respectively depict in detail the timelines for the two different mechanism design models.}

In timeline A, the mechanism is implemented without delay: the announcement of its existence, the decision to participate and what to play happen almost concurrently so that each agent only experiences one bout of news utility (dubbed \emph{surprise effect}) coming from comparing the pre-mechanism beliefs in the consumption and money dimensions pinned down by her outside option with the new (degenerate) beliefs induced by the realization of the mechanism outcome. 

In timeline B the participation and play stages of the mechanism happen without delay but the mechanism outcome is realized with delay. Besides the surprise effect the agent experiences now a second bout of news utility coming from comparing the actual realization of the mechanism outcome with the distribution induced by her play decision and the environment (dubbed \emph{realization effect}). She takes this into account in expectation at the play stage. This results in lower interim utility because of loss aversion: delaying the outcome after a play decision hurts the agent because bad news hurt more in expectation than good news elate. 

Finally, in timeline C there is a delay between the participation stage and the play stage, besides the delay between the play stage and the moment the mechanism outcome is realized. In this case the agent's time inconsistency becomes observable as different selves with different objectives decide on participation and play. The play-self doesn't take into account the surprise effect of the participation self whereas the latter takes into account the future behavior of the play-self. This wedge between participation and play can be fruitfully used by the designer in certain situations as she may be able to exploit the play-self better once she is locked-in after deciding to take part in the mechanism. This comes at a cost though: in optimal mechanisms the participation-self, being sophisticated and anticipating her future decisions, may need to be subsidized in comparison to the other timelines. 

Say that a direct mechanism is \emph{incentive compatible} if revealing own private information is an equilibrium for the respective timeline of the mechanism. Similar to the classical setting, monotonicity conditions are key in characterizing incentive compatibility and the resulting expected transfers. But due to news utility, monotonicity applies to modified interim \emph{perceived valuations} instead of the interim intrinsic valuations from the classical setting. As a consequence, incentive compatible allocation rules only determine \emph{perceived expected transfers} to the designer, up to type-independent constants. 

Say that an incentive compatible mechanism is \emph{individually rational} if it gives an equilibrium participation utility higher than the outside option to every agent. Whereas the same self decides about incentive compatibility and individual rationality whenever there is no delay between participation and play decisions (timelines A and B), different selves decide on them when there is a delay between the two decision stages (timeline C).

For allocation rules whose incentive compatibility is unproblematic in the classical model, loss aversion can lead to failure of incentive compatibility whenever it is high enough and the timeline features delays. We illustrate this in the case of the ex-post efficiency rule for public good provision with symmetric agents and private information concerning the intrinsic type.\footnote{Ex-post efficiency means the designer would like to maximize the welfare of the agents under complete information.} We also show how incentive compatibility may be restored in that setting under certain conditions on the distribution of the intrinsic valuation of the public good, whenever the number of agents is high enough. Intuitively, with a large number of agents the law of large numbers kicks in and the ex-post efficiency rule implies small news utility costs for delays as the probability of provision becomes either very high or very low.

For the screening model we show that a timeline without any delays (timeline A) is optimal whenever the private information of the agent, her loss aversion, is symmetric across the two dimensions, consumption and money. Furthermore, a timeline without delays between participation and play decision is always weakly better than one with delays between the two decision stages (timeline C). This ranking relies on the assumed symmetry for the loss aversion across dimensions, which implies a relatively small negative surprise effect in the money dimension at the moment of the participation decision, as well as on the fact that the subsidy for the individual rationality constraint in timeline C is non-negligible. 

For the multiple agent model where private information concerns intrinsic valuations we study the optimality of timelines in the case of auctions with symmetric agents. We show that, analogous to the screening model, timeline B is never optimal as it is dominated by a timeline A. We show by example that timeline C may be optimal whenever the lowest intrinsic type is strictly higher than the highest utility of the outside option of an agent. Intuitively, this is precisely the case when the agent has maximal participation incentives. This lowers the subsidy needed to overcome the incentive wedge between the participation-self and the reporting-self. When compared to A, timeline C has an improved incentive compatibility property due to two effects: 1) the missing negative surprise in the money dimension, and 2) the designer can condition payments on the uncertainty facing the agents after their declared type in such a way as to give better incentives for truthtelling. The second effect is costly for each agent due to loss aversion, but may result in higher revenue overall due to improved incentive compatibility. Overall, when the loss aversion in the money is high enough, so that the first effect is strong, timeline C may dominate a timeline A. In the same example and for the case that loss aversion in the money dimension is low enough we show that a timeline A may remain optimal.




The rest of the paper is organized as follows. In the next subsection we introduce the model for a single agent which is the building block for both of the mechanism design models we consider. Subsection \ref{sec:relatedlit} comments on related literature. Section \ref{sec:screening} contains the one-agent screening model for loss aversion followed by Section \ref{sec:multiagent} which considers the multi-agent model where agents face uncertainty about the intrinsic types of their opponents. Proofs for the main results of the paper are relegated in the appendix. The online appendix comments on the revelation principle in our setting and relaxes the assumption of a degenerate outside option in the money dimension. It also contains several other applications of the multi-agent model to other classical Bayesian mechanism design settings.  

\subsection{Preferences}\label{sec:prefs}




This subsection explains in detail the preferences and the decision procedure of the agents. 
Agents experience intrinsic utility from actual consumption and payments as well as news utility from changes in beliefs. 

\paragraph{Intrinsic utility.}

The agent derives intrinsic utility from consumption of a profile $a$ of consumption goods coming from a set of physical allocations $\mathcal{A}$ which is a closed, connected subset of an Euclidean space $\R^d,d\ge 1$ as well as from a monetary transfer $t$. Moreover, the utility of the agent from the pair of consumption goods and monetary payment $(a,t)$ depends also on a parameter which we call her \emph{intrinsic type} $\theta$ and which comes from a closed interval of $\R$ denoted $\Theta = [\underline{\theta}, \bar \theta]$. We assume that intrinsic utility for each agent with type $\theta$ is quasilinear and of the form
\begin{equation}\label{eq:intrinsic}
V(a,\theta,t) = v(a)\theta-t.
\end{equation}
$v:\mathcal{A}\ra\R_+$ is assumed differentiable. Let in the following $\Delta(\A\times \Theta\times\R)$ denote the set of Borel probability distributions over $\A\times \Theta\times\R$.

We assume the agent conforms to Expected Utility in the intrinsic part of her utility. That is, if the uncertainty over $u=(a,\theta, t)$ is captured by a distribution $G\in \Delta(\A\times \Theta\times\R)$ then the agent experiences \emph{expected intrinsic utility} of the form 
\[
\E_{u\sim G}[V(u)].
\]
Hereafter, $u\sim G$ means the random variable $u$ is distributed according to the distribution $G$.
\paragraph{News utility.}
Besides intrinsic utility, the agent experiences news utility whenever her beliefs about the realization of $u$ objectively change from $H$ to some $G\in \Delta(\A\times \Theta\times\R)$. News utility is experienced in two dimensions: in the consumption dimension (upper index $g$ in the following, $g$ stands for \emph{good}), and in the monetary dimension (upper index $m$ in the following). For any $G\in \Delta(\A\times \Theta\times\R)$ denote by $G^g$ the distribution of $v(a)\theta$ induced by $G$. This depends only on the marginal distribution of $G$ on $\A\times \Theta$. Moreover, let $G^m$ be the marginal distribution of $G$ on the monetary payments. Whenever the agent's belief changes from $G$ to $H$ she experiences news utility in the dimension $j=g,m$ given by 
\begin{equation*}
\mathcal{N}^j(G^j|H^j) = \mu^j\int_0^1 \xi^j(c_{G^j}(p)-c_{H^j}(p))dp.
\end{equation*}
Here for any real-valued distribution $F$, $c_F(p)$ is the $p$-percentile of $F$, $p\in (0,1)$.\footnote{For any $p\in (0,1)$ the \emph{$p$-percentile $c_F(p)$ of $F$} is determined by the conditions $F(c_F(p))\geq p$ and $F(c)<p$ for any $c<c_F(p)$.}
Moreover, 
\begin{equation}
\label{eq:valkt}
\xi^j(y) = 
\begin{cases} y &\mbox{if } y \geq 0 \\
\lambda^j y &\mbox{if } y < 0 
\end{cases}
\end{equation}
is a value function of Kahneman-Tversky type (\cite{kt}).\footnote{This percentile-per-percentile comparison first appeared in \cite{kr3}. See \cite{pagel} for alternative specifications of the news utility which can incorporate correlation between dimensions.} Here $\lambda^j>1$ is the loss aversion parameter of the agent in the dimension $j$. A negative change causes the agent to experience a disutility greater in magnitude than the elation caused by a positive change of the same size.

$\mu^j>0$ is the agent's relative weight on the news utility in dimension $j$. Overall, news utility of the change from $H$ to $G$ is given through the sum of the news utilities of the two dimensions. 

\[
\mathcal{N}(G|H) = \mathcal{N}^g(G^g|H^g) + \mathcal{N}^m(G^m|H^m).
\]

The timelines we have in mind are the following. First, the agent is offered a menu of lotteries $\C$ which is a subset of $\Delta(\A\times \Theta\times\R)$. We assume this is a non-empty compact set of $\Delta(\A\times \Theta\times\R)$ where we have equipped the latter with topology of weak convergence of probability measures (this fits both mechanism design models below). Then she picks a lottery from $\C$ and finally the outcome of the lottery is realized. There may exist delays between the moment the menu $\C$ is offered to the agent and the moment she chooses from $\C$ as well as the moment the uncertainty from the lottery she picked from $\C$ is realized. These delays may be due to technological constraints or they may be introduced through the outside party, call it designer, which designs the menu $\C$.\footnote{Some typical examples of technological constraints comprise settings where communication takes time or where delivery of payments/goods or production of a good whose consumption value is uncertain takes time.} \footnote{In the mechanism design settings we consider not every menu $\C$ out of $\Delta(\A\times \Theta\times\R)$ is feasible. In particular, we don't allow the designer to randomize so that the subjective randomness the agents face is only due to the environment.}
We assume that there is no discounting of time whenever a delay is present.\footnote{It is not hard to introduce discounting to the model but discounting doesn't yield any additional deep insight besides making the model much more cumbersome.} 

The timelines the agent may face depending on the timing of delays are depicted in Figure \ref{fig:oneagent}.\footnote{Additional timelines featuring delay between the announcement of the existence of $\C$ and the moment the agent is required to decide on whether to accept $\C$ or not are equivalent to existing ones under the no-discounting assumption by the same argument as we show for timelines C and D below (see Proposition \ref{thm:oneagenthelp}). Those timelines would be relevant in a model of \emph{`deciding when to decide'} which is outside the scope of this paper.} 

\begin{figure}[H]
  \centering
   \includegraphics[width=0.58\textwidth]{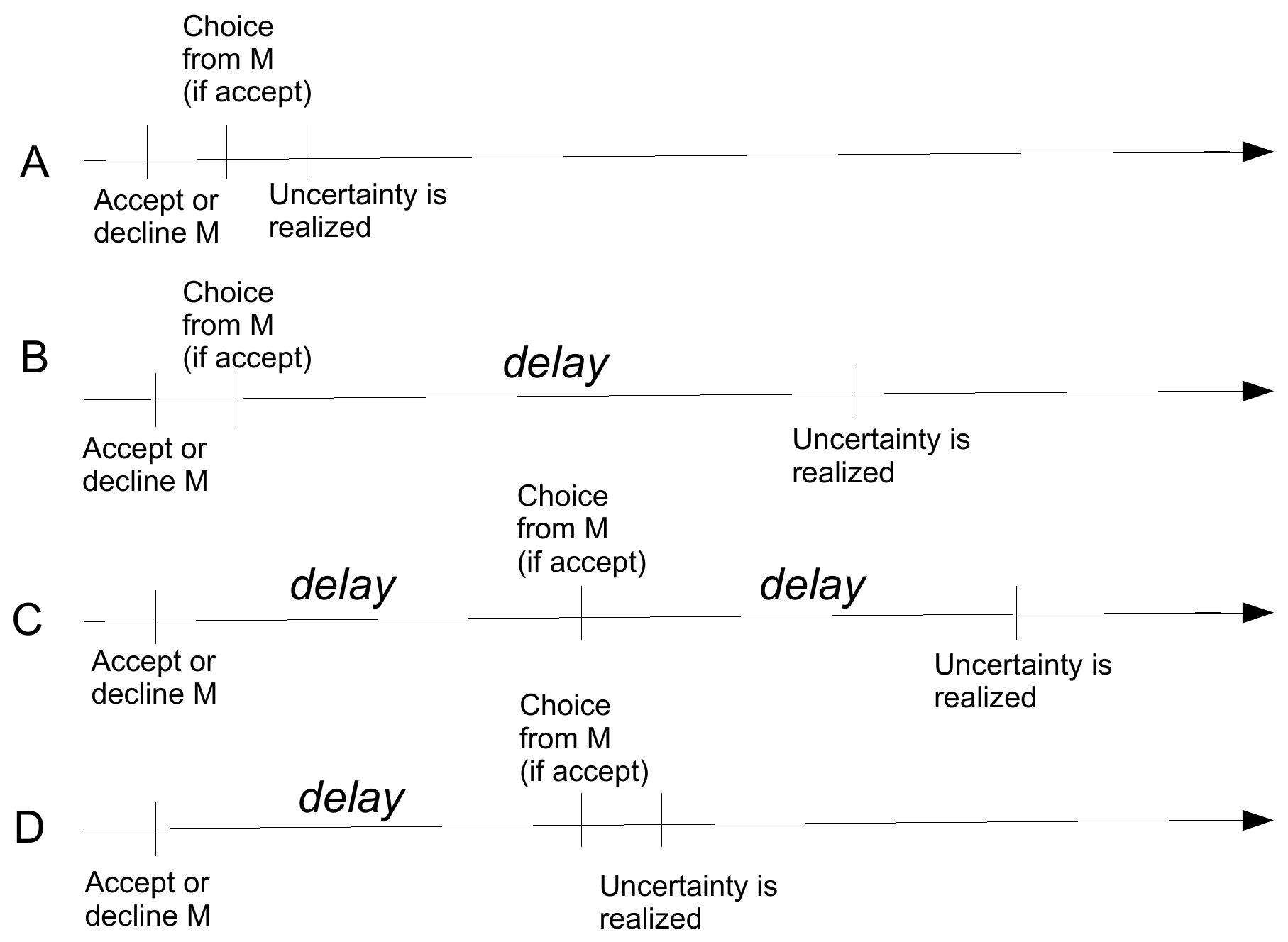}
    \caption{Different timelines.}\label{fig:oneagent}
\end{figure}


In the following we also assume that the agent has an exogenously given belief over $\A\times \Theta\times\R$ given by $F_0$. This is what she expects to happen if she is not notified of the option of choosing from $\C$. In mechanism design settings below $F_0$ is determined by the outside option of the agent.

To calculate the overall utility of an agent we impose the following assumptions on the agent's behavior.

\paragraph{\emph{Assumption 1:}} The agent is sophisticated, Bayesian and forward-looking. That is, she takes into account the optimal behavior of future selves, uses Bayes rule to update beliefs whenever possible and is indifferent to the welfare of past selves.
\vspace{2mm}\\
This means that when deciding whether to accept the menu $\C$ the agent takes into account the actual choice she will make from $\C$. Moreover, at each decision moment after a delay no past intrinsic or news utility is taken into account. We note here that Assumption 1 doesn't necessarily imply a temporal coordination of selves in the sense that the self picking from $\C$ needs to break ties in favor of the self who decides to accept or reject $\C$.\footnote{An idea of coordination among selves to give a past self a higher utility underlies the PPE concept in \cite{kr1}.} 

\paragraph{\emph{Assumption 2:}} News utility is produced only from \emph{objective} sources, that is, there is no self-production of news utility. In particular, the agent conforms with Expected Utility at each moment in time with respect to any subjective randomization device.\footnote{The related PPE concept used in \cite{kr1} (see also \cite{kr3}) to model agents who experience expectation-based loss aversion allows for the possibility of self-production of news utility and therefore adds an additional constraint to the maximization problem of the agent. This may lead to existence and characterization problems (see \cite{koeszegi}). These are excluded by our assumptions.} 
\vspace{2mm}\\
Under this assumption it is unproblematic to assume that the agent doesn't possess a randomization device when picking from $\C$.

\paragraph{\emph{Assumption 3:}} At any moment in time, if no delay is present there is at most one new news utility term. It comes from comparing beliefs before an objective source of news with those after the objective source of news.
\vspace{2mm}\\
Assumptions 2 and 3 imply that there are at most two instances of production of news utility in the above timelines: 

1) when the menu $\C$ is presented to the agent if she decides to accept it \emph{(surprise effect)}, as well as

2) when the uncertainty of the lottery she picked from the menu is realized \emph{(realization effect)}. 

Assume the agent ultimately chooses $F\in \C$. Then the surprise effect in timelines B,C,D corresponds to experiencing $\mathcal{N}(F|F_0)$.  The realization effect for timelines B,C and D corresponds to experiencing $\mathcal{N}(u|F)$ whenever $u\in \A\times\Theta\times\R$ is the realization of $F$. 

In timeline A the two news utility effects coincide and there is only one news utility term comprised of $\mathcal{N}(u|F_0)$ whenever $u\in \A\times\Theta\times\R$ is realized. This is because timeline A stands for the case where the decisions of whether to accept $\C$, what lottery to pick out of $\C$ and the realization of the resulting uncertainty all happen without delay and almost concurrently.


\paragraph{\emph{Assumption 4:}} Each agent in a decision moment takes into account future news utility terms in expectation by weighting them with the belief induced by her actual decision be it a decision \emph{on or off}-equilibrium path.\footnote{Just as in the case of Assumption 2 this is in stark difference to the PPE solution concept in \cite{kr1} and \cite{kr3}.}
\vspace{2mm}\\
This implies for an agent in timelines $B,C,D$ who accepted the menu $\C$, that the expectation of her news utility from the realization effect whenever she picks $F\in \C$ is given by $\E_{u\sim F}[\mathcal{N}(u|F)]$. We call such a term an \emph{expected news utility term}. 

\paragraph{Decision procedure of a single agent.}

In timeline A the expected news utility from accepting the menu $\C$ and choosing $F\in \C$ enters the overall decision utility of the agent as $\E_{u\sim F}[\mathcal{N}(u|F_0)]$.

For the case that the agent doesn't accept menu $\C$ she experiences utility $$O(F_0) = \E_{u\sim F_0}[V(u)]+\E_{u\sim F_0}[\mathcal{N}(u|F_0)],$$
regardless of whether the realization of $F_0$ happens with a delay or not.\footnote{$O$ stands for \emph{outside option}.} \footnote{She experiences only one news utility term given by $\mathcal{N}(u|F_0)$ whenever $u$ is realized and takes expectation of it by weighting with $F_0$. This is because $F_0$ is not a surprise, i.e. it is expected by the agent at the beginning of time.} For the case that $F_0$ is degenerate, say puts probability one on $u_0$, we have $O(F_0) = V(u_0)$.

Whenever the agent accepts the menu, she experiences in timeline A the decision utility 

\[
\max_{F\in \C}\E_{u\sim F}[V(u)]+\E_{u\sim F}[\mathcal{N}(u|F_0)].
\]
She accepts $\C$ if and only if this utility is higher than $O(F_0)$.

In timeline B she accepts $\C$ if and only if her decision utility from $\C$ given by 
\begin{equation}\label{eq:oneagent1}
\max_{F\in \C}\E_{u\sim F}[V(u)]+\E_{u\sim F}[\mathcal{N}(u|F_0)] + \E_{u\sim F}[\mathcal{N}(u|F)]
\end{equation}
is higher than $O(F_0)$. She then picks an $F$ from $\C$ where the maximum is attained.

In timelines C and D different selves of the agent with distinct perspectives decide on whether $\C$ should be accepted and then on the lottery picked out of $\C$. For the case that the agent has accepted $\C$ her decision utility from choosing out of $\C$ is

\[
\max_{F\in \C}\E_{u\sim F}[V(u)]+\E_{u\sim F}[\mathcal{N}(u|F)].
\]

Being sophisticated, she then accepts $\C$ if and only if her decision utility $\E_{u\sim F}[V(u)]+\E_{u\sim F}[\mathcal{N}(u|F_0)] + \E_{u\sim F}[\mathcal{N}(u|F)]$, evaluated at the $F$ she expects to choose out of $\C$, is higher than $O(F_0)$. If there are multiple $F\in \C$ which are optimal we assume she breaks ties deterministically and that her self at the moment of deciding whether to accept $\C$ anticipates the tie-breaking correctly. This is in line with Assumption 1 above.

Formally, we assume the following condition about possible tie-breaking.

\paragraph{Deterministic tie-breaking.} Whenever the agent is indifferent between accepting $\C$ and rejecting, she accepts it. Whenever the agent is indifferent between distinct elements of $\C$ at the choice-out-of-menu stage she breaks ties deterministically.\footnote{Deterministic tie-breaking is in line with Assumption 2 above. We assume it here so that we can focus on the behavioral features of the model rather than technicalities. 
} \footnote{All settings considered in this paper correspond to menus whose elements are parametrized through a compact one-dimensional interval so that it is easy to write down deterministic tie-breaking rules.}
\vspace{2mm}\\
Our Assumptions together with deterministic tie-breaking ensure a full characterization of behavior. 
It is easy to see that the behavior in timeline D is the same as in timeline C. The following Proposition registers this property as well as another simple one which has important consequences in mechanism design settings.

\begin{proposition}\label{thm:oneagenthelp}
1) Expected news utility terms of the form $\E_{u\sim F}[\mathcal{N}(u|F)]$ are non-positive and vanish if and only if $F$ puts unit mass on a single element $u_0$.

2) The behavior of the agent is the same in timelines C and D. 
\end{proposition}

The proof of the equivalence of C and D is contained in the text above. It relies crucially on the no-discounting assumption as well as on Assumption 2.\footnote{In the mechanism design settings below there is no uniform result as to which of timelines C,D is better once discounting is allowed. Intuitively, discounting applies in timeline C to both intrinsic utility and news utility at the choice-from-menu stage whereas it applies to neither in timeline D.} Part 1) is an implication of loss aversion. Because of loss aversion percentile comparisons between the realized $u$ and the percentiles of the `reference' distribution $F$ are weighted asymmetrically depending on whether they correspond to a loss or gain; losses in the dimension $j$ get a weight of $\mu^j\lambda^j$ whereas gains only of $\mu^j$ in units of intrinsic utility. Averaging out across realizations of $u$ results in a negative \emph{expected} news utility effect. 

Due to Proposition \ref{thm:oneagenthelp} we identify timelines C and D in the rest of the paper. 


\subsection{Related Literature}\label{sec:relatedlit}


This paper connects to different strands of the mechanism design literature as well as of the applied behavioral literature. The focus on optimal timeline choice for preferences which are sensitive to the timing of announcements and additionally feature loss aversion seems new in the literature.

\cite{suspense} analyzes the optimal way to disclose information to an agent whose preferences depend on the path of the belief change. In their model the agent has preference for late resolution of uncertainty as she likes to experience suspense and surprise. In our model the agent exhibits preference for early resolution of uncertainty due to loss aversion. Moreover, the goal of the designer in \cite{suspense} is to maximize welfare of the agent whereas we focus on profit maximization. Finally, their model features a fixed timeline whereas we also study the optimal choice of the timeline in our model.

There are by now several screening models where sophisticated agents exhibit loss-aversion. \cite{carbajal14} proposes a screening model of reference dependent and loss averse consumers where the reference point is non-stochastic and depends linearly on the private information of the consumer. 
\cite{hahn} consider price discrimination with loss-averse consumers. Similar to the model proposed in Section \ref{sec:screening} they assume the buyers don't know their intrinsic valuation at the moment they face the menu of bundles the monopolist offers. They assume that the reference point is the menu of bundles the monopolist offers and work with ex-post participation and incentive compatibility constraints. Loss aversion parameters are known by the monopolist in their model. Our model assumes the monopolist has imperfect information about loss aversion parameters. Moreover, our model introduces and studies the issue of designing the optimal screening timeline which is missing from both papers mentioned.

The applied behavioral literature offers several models where a designer, say a monopolist or a firm in a competitive market, screens on behavioral features of the agents. \cite{es1} and \cite{es2} consider a designer who faces agents who may hold potentially incorrect beliefs about their future utility and screens respectively on the level of sophistication or on the level of optimism of the agent. \cite{hk2} offer a related model and study welfare consequences of screening for the sophistication level of agents. \cite{hk1} and \cite{efs} offer models in various settings where a designer screens agents according to their present-bias level. In all of these models one of the main reasons the agents are dynamically inconsistent is that they may be naive about their future behavior whereas in our paper the agent is sophisticated and dynamic inconsistency arises directly from an agent's preferences and not through incorrect beliefs about future behavior. 

\cite{heidhues} considers a monopolistic setting where there is complete information about agent's preferences, but the monopolist can commit to draw the price from an ex-ante designed and announced distribution: in period 1 the agent observes the price distribution and forms expectations, while in period 2 it observes the drawn price and decides whether to buy or not. The optimal price distribution is non-degenerate. In contrast, in this paper the designer is assumed to reveal all details of the mechanism in a single step at the beginning of the game and we work with interim (as opposed to ex-post) participation constraints.\footnote{In fact, the result in \cite{heidhues} ceases to hold if the agent is given the chance to decide about the purchase before seeing the distribution of prices in period 1. Ex-post participation constraints are less appropriate in settings where the value of the product to the consumer is revealed only upon consumption and the agent can commit to the mechanism before consumption.} Nevertheless, since in our multi-agent model uncertainty persists even with deterministic mechanisms, we get a related result to theirs in our setting: it may be optimal for the designer to not insure the agent against future uncertainty. Finally, we stress that their paper focuses on PPE types of equilibria which can't arise in this paper.

To the best of our knowledge, \cite{eisenhuth} is the first paper considering a full-fledged mechanism design model of auctions with preferences which exhibit expectation-based loss aversion. He uses the equilibrium concept from \cite{kr2} (CPE) and solves for the optimal symmetric auction with symmetric bidders. His environment is most similar to timeline B in this paper.\footnote{Similarly to our results in the online appendix \cite{eisenhuth} establishes that optimal auctions for CPE preferences are all-pay with a reference price. He also considers a model of wide-bracketing, which is assumed away in our model due to the separability assumption across the two dimensions, consumption and money.} \cite{herwegetal} introduce the same CPE preferences in the classical principal-agent model with moral hazard and show how the optimal contract is much simpler than in the classical model. In a new paper, \cite{benkert} considers the optimal mechanism problem for the bilateral trade model where both buyer and seller behave according to CPE. He solves for the optimal mechanisms for a special distribution of types. Within the same timeline our paper additionally looks at other topics from classical Bayesian mechanism design: auctions, public goods. Moreover, we also address the issue of the optimal timing of the realization of the mechanism, which doesn't occur in either of the above-mentioned works.


This paper contributes to the emerging literature on strategic interaction of agents whose utility depends on their beliefs about present and future consumption and money transfers. 
Relatedly, \cite{dato} characterize existence properties of strategic equilibrium based on the preferences of \cite{kr1} (their PE and PPE concepts) and \cite{kr2} (their CPE concept) in finite normal-form games and focus mostly on existence and uniqueness as well as characterizing when equilibrium play is the same under classical and expectation-based loss averse preferences. The fact that ex-post efficiency in the public goods setting with timelines B,C,D may fail incentive compatibility (see subsection \ref{sec:multiagentexpostpublic}) is a reflection of the same non-existence phenomenon they identify in their model.

\section{A screening model of loss aversion}\label{sec:screening}

This section offers a tractable model of screening with news utility agents whose private information concerns their loss aversion parameters.
Naturally, assuming that private information of an agent includes multiple parameters of her behavioral preferences leads to a multidimensional screening problem. Here we avoid the technical difficulties of multidimensional mechanism design by assuming symmetry for loss aversion in the money and good dimension and that the designer knows all news utility parameters but one: the loss aversion parameter. We focus on a model of monopolistic screening where the buyers don't know the realization of their intrinsic utility at the moment of participation decision and choice of contract.\footnote{It is possible to construct a more complicated model with signals which partly reveal the incomplete information about $\theta$ after or before the contract is signed. If the signals are public knowledge however, the more general case is easily reduced to a model similar to the one in this section.}   

Real life settings approximated by this model would be buying tickets to a concert from an unknown band, buying a book from an unknown author, or vacationing in an unknown destination, etc.


\subsection{Set Up}\label{sec:screeningsetup}

We assume the designer is a monopolist producing non-negative quantities of a good denoted by $q$ at a fixed marginal cost $c>0$.\footnote{This can be relaxed to a weakly increasing marginal cost $c$ and none of the results would qualitatively change. We keep a constant marginal cost throughout for ease of exposition.} A buyer has intrinsic utility from a contract $(q,t)$ of the form
\[
v(q)\theta-t.
\]
Here $\theta$ is the intrinsic value of the good. 
$v$ fulfills standard assumptions: $v(0)=0,v\geq 0, v'>0,v''<0$ and $\lim_{x\ra 0}v'(x) = +\infty$. Moreover, for purely technical reasons we additionally require a weak growth condition on $v$: there exists some $p>1$ so that it holds $v'(x)\geq \frac{K}{v^{p-1}(x)}$ for some $K>0$.\footnote{This weak condition is crucial for our proof of existence of an optimal mechanism in timeline C. It says that the marginal utility $v'$ doesn't fall too fast with $x$.} 

Figure \ref{fig:screening} below depicts the relevant timelines with news utility. We assume the timeline $T\in \{A,B,C\}$ is either given through technological constraints or it is chosen by the designer.
\begin{figure}[h!]
  \centering
    \includegraphics[width=0.65\textwidth]{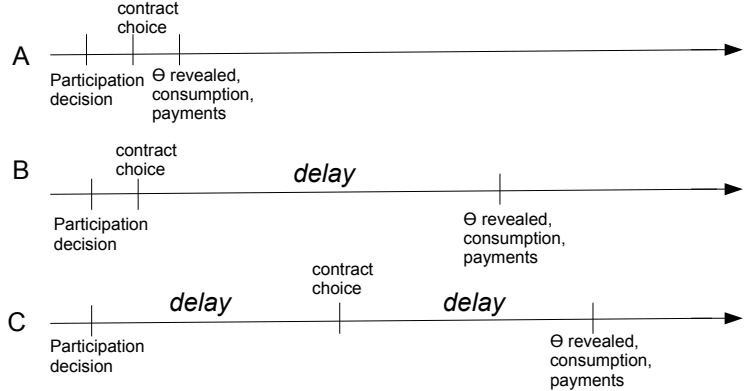}
      \caption{Relevant timelines for the monopolistic screening model.}\label{fig:screening}
\end{figure}
The informational assumptions of the screening model are as follows.

\paragraph{Assumption (S)}  Intrinsic utility $\theta$ is distributed according to a probability distribution $F$ with bounded support over the non-negative numbers. $F$ is common knowledge and it fulfills $m:=\E_{\theta\sim F}[\theta]>0$ (on average the good yields positive intrinsic utility). All buyers put weight $\mu^g = \mu^m = 1$ on news utility and have $\lambda^g=\lambda^m=\lambda$. The distribution of $\lambda$ has a continuously differentiable and strictly positive density $g$ on $[1,\bar\lambda]$ $(2\ge \bar\lambda>1)$ and its c.d.f. $G$ is common knowledge.\footnote{As \cite{mr} establish, $\lambda\le 2$ is a necessary requirement for the preferences of the agent to respect first-order stochastic monotonicity in the money dimension. They don't consider different timelines as here but their results about choice over monetary lotteries still hold qualitatively in our setting of timelines B,C (and therefore also D).} 
\vspace{4mm}\\
Denote for future reference 
$M:=\E_{\theta,s\sim F}[(\theta-s)\textbf{1}_{\{\theta\geq s\}}]$. 
It is easy to see that $M<m$.

The outside option which determines the pre-mechanism beliefs of each buyer type consists of a zero utility: zero amount of good and zero transfers to the designer are expected in the absence of any mechanism. 

We focus in this section on the case of deterministic mechanisms, i.e. we assume for simplicity in exposition that the monopolist doesn't possess a randomization device. We comment in the end of this section on how the results change with randomized mechanisms. 

Mechanisms consist of menus $\mathcal{C} = \{(q(\lambda),t(\lambda))\}_{\lambda\in [1,\bar\lambda]}$ the monopolist offers. This is a set of contracts indexed by loss aversion specifying the quantity of the good $q$ and the price of that good at that quantity. In the framework of subsection \ref{sec:prefs} the buyers are agents facing a menu of lotteries over $\Delta([0,\infty)\times\Theta\times\R)$ with the property that the marginals over the quantity $q$ and payment $t$ are degenerate. The lotteries in the menu are indexed by $\lambda\in [1,\bar\lambda]$.

In the remaining part of this section we characterize individual rationality and incentive compatibility for all timelines and finally look at the optimal timeline choice for the designer. 

\subsection{Incentive Compatibility and Individual Rationality}\label{sec:screeningicir}

\subsubsection{Incentive Compatibility}

\paragraph{Timeline A:} It is without loss of generality for optimal mechanisms to consider only payment schedules $t\ge 0$.


If the agent has decided to participate and choose bundle $(q(\hat\lambda), t(\hat\lambda))$ then the news utility from the comparison with the pre-mechanism expectations is given by 
\begin{align}\label{eq:surprisescreening}
v(q(\hat{\lambda}))m-\lambda t(\hat{\lambda}).
\end{align}
She also experiences consumption utility $v(q(\hat{\lambda}))\theta-t(\hat{\lambda})$, which at the decision moment is in expected utility terms 
\[
v(q(\hat{\lambda}))m-t(\hat{\lambda}).
\]
Gathering the terms together, we see that utility of a buyer of type $\lambda$ from declaring $\hat{\lambda}$ is 
\[
2mv(q(\hat{\lambda}))-(1+\lambda)t(\hat{\lambda}).
\]
Denote 
\begin{equation}
\Gamma^A(\lambda) = \frac{2m}{1+\lambda}.
\end{equation}
We call $\Gamma^A(\lambda)$ the \emph{A-virtual type} of the buyer. The term in the denominator reflects the negative surprise effect in the money dimension. The numerator reflects the surprise effect in the good dimension. The \emph{A-virtual type} is a decreasing function of the agent's loss aversion parameter.

Incentive compatibility is characterized by
\[
\argmax_{\hat{\lambda}\in [1,\bar\lambda]}\left\lbrace(1+\lambda)\left[\Gamma^A(\lambda)v(q(\hat{\lambda}))-t(\hat{\lambda})\right]\right\rbrace = \lambda.
\]

\paragraph{Timeline B:}
Again it is without loss of generality for optimal mechanisms to consider only payment schedules $t\ge 0$.

In this timeline if the buyer of type $\lambda$ has decided to choose bundle $(q(\hat{\lambda}),t(\hat{\lambda}))$, she will first experience news utility from the comparison with the pre-mechanism expectations just as in timeline A.

When the buyer learns her draw of the material valuation $\theta$ she experiences news utility in the good dimension from comparing the outcome to her previous belief $v(q(\hat{\lambda}))s$ where $s$ is distributed according to $F$:
\begin{align*}
v(q(\hat{\lambda}))\int (\theta-s)\textbf{1}_{\{\theta\geq s\}}+\lambda(\theta-s)\textbf{1}_{\{\theta< s\}}F(ds).
\end{align*}
Due to the delay, she takes the expectation of this expression w.r.t. $\theta$ at the moment she decides whether to participate or not, so that her expected news utility from realization is given by the following expression
\[
(1-\lambda)Mv(q(\hat{\lambda})).
\]
There is no news utility term in the money dimension when $\theta$ is revealed as there is no uncertainty in the money dimension once a bundle has been chosen by the agent. 

Taking into account the expectation of the intrinsic utility $v(q(\hat{\lambda}))\theta-t(\hat{\lambda})$ we see that decision utility of a buyer of type $\lambda$ from choosing the contract corresponding to $\hat{\lambda}$ is 
\[
\left[2m+(1-\lambda)M\right]v(q(\hat{\lambda}))-(1+\lambda)t(\hat{\lambda})
\]
Denote 
\begin{equation*}
\Gamma^B(\lambda) = \frac{2m+(1-\lambda)M}{1+\lambda}.
\end{equation*}
We call this the \emph{B-virtual type} of the agent. It is a decreasing function of her loss aversion parameter. The denominator reflects the surprise effect in the money dimension while the numerator reflects both the surprise as-well-as the realization effect in the consumption dimension.

\emph{Incentive Compatibility} is characterized by
\[
\argmax_{\hat{\lambda}\in [1,\bar\lambda]}\left\lbrace(1+\lambda)\left[\Gamma^B(\lambda)v(q(\hat{\lambda}))-t(\hat{\lambda})\right]\right\rbrace=\lambda.
\]

\paragraph{Timeline C.}

At the moment of contract choice the agent experiences no surprise effect as the contract doesn't constitute news anymore but she still  takes into account the realization effect in expectation. Given this, \emph{Incentive Compatibility} is characterized by

\[
\argmax_{\hat{\lambda}\in [1,\bar\lambda]}\left\lbrace[m+(1-\lambda)M]v(q(\hat \lambda))-t(\hat\lambda)\right\rbrace = \lambda.
\]

Denote $\Gamma^C(\lambda)=m+(1-\lambda)M$ the \emph{C-virtual type}. It includes the expected future news utility from the realization of the mechanism as well as the expected value of consumption. It is decreasing in the loss aversion parameter $\lambda$. 

Standard methods yield then the following characterization of incentive compatible mechanisms. Here we call an allocation rule $q:[1,\bar\lambda]\ra \R_{+}$ is implementable if there exists an incentive compatible mechanism $\C$ with allocation rule $q$. 

\begin{proposition}\label{thm:screening-ic}
1) An allocation rule $q:[1,\bar\lambda]\ra \R_{+}$ is implementable if and only if $q(\cdot)$ is non-increasing.

2) (Mirrlees Representation) For any implementable $q$ the corresponding payments $t:[1,\bar\lambda]\ra \R$ are given up to a type-independent constant by the following Mirrlees representations. 
\begin{align*}
&(A)\quad 
t^A(s) = \Gamma^A(s)v(q(s))-2m\int^{\bar\lambda}_s\frac{v(q(t))}{(1+t)^2}dt,
\\&(B)\quad t^B(s) = \Gamma^B(s)v(q(s))-2(m+M)\int^{\bar\lambda}_s\frac{v(q(t))}{(1+t)^2}dt,\\
&(C)\quad t^C(s) = \Gamma^C(s)v(q(s))+M\int_1^{s}v(q(\sigma))d\sigma.
\end{align*}
\end{proposition}


The type-independent constant not depicted in part 2) of the Proposition is a fixed payment to the monopolist (or transfer to the agents) which is independent of the private information of the agents. In an optimal mechanism its value is determined by the individual rationality requirement. 

\subsubsection{Individual Rationality} 

The following Proposition gives the individual rationality characterization for incentive compatible mechanisms where the payment schedules $t$ are non-negative.\footnote{This is without loss of generality for optimal mechanisms as we show in the appendix.}

\begin{proposition}\label{thm:screening-ir}
Fix a timeline $T\in \{A,B,C\}$. An incentive compatible contract $\mathcal{C} = \{(q(\lambda),t^T(\lambda))\}_{\lambda\in [1,\bar\lambda]}$ is individually rational for timeline $T$ if the following respective sets of inequalities are satisfied.

\begin{itemize}
\item For timelines $T=A,B$
\[
\Gamma^T(\lambda)v(q(\lambda))-t^T(\lambda)\geq 0,\quad \lambda\in [1,\bar\lambda].
\]
\item For timeline $C$ 
\[
\Gamma^B(\lambda)v(q(\lambda))-t^C(\lambda)\geq 0,\quad \lambda\in [1,\bar\lambda].
\]
\end{itemize}

\end{proposition}

The expressions for timelines A, B are similar to ones from classical models except for the fact that one has to use modified virtual types which take into account news utility effects. 

For timeline $C$, at the participation stage the agent anticipates that she'll be truthful later, if she accepts an incentive compatible menu $\C$. In all incentive compatible mechanisms in timeline C the payment schedule $t(\lambda)$ is weakly decreasing in $\lambda$.\footnote{This follows immediately from the Mirrlees representation coupled with the envelope theorem.} Given this, and the fact that the monopolist can always offer the bundle $(0,0)$ to any type it follows that the optimal contract will never feature a net subsidy $t(\lambda)<0$ for any type $\lambda\in[1,\bar\lambda]$.\footnote{If this were not true then it would hold $t(\lambda_0)<0$ for every $\lambda_0>\lambda$. Going over to $t(\lambda_0) = 0$ and $q(\lambda_0)=0$ for all $\lambda_0\ge \lambda$ preserves incentive compatibility but increases profits.}

It follows that in case of acceptance her utility if she is of type $\lambda$ is given by 
\[
V(\lambda) = \left[2m+(1-\lambda)M\right]v(q(\lambda))-(1+\lambda)t(\lambda).
\]
Namely, the agent experiences the following utility items: news utility from accepting the mechanism, intrinsic \emph{expected} utility from future consumption and finally \emph{expected} future news utility from the realization effect.

We show in the appendix that $V(\lambda)$ can be rewritten as
\[
V(\lambda)=(1+\lambda)\left[\Gamma^B(\lambda)-\Gamma^C(\lambda)\right]v(q(\lambda))-(1+\lambda)f-(1+\lambda)M\int_1^{\lambda}v(q(s))ds. 
\]
where $f$ is a type-independent payment. 
The individual rationality requirement can then be written as 

\[
\left[\Gamma^B(\lambda)-\Gamma^C(\lambda)\right]v(q(\lambda))-f-M\int_1^{\lambda}v(q(s))ds\geq 0,\quad \lambda\in [1,\bar\lambda]. 
\]

The fact that different selves decide on participation and bundle choice creates a `wedge' between individual rationality and incentive compatibility requirement. A measure for this discrepancy is precisely the multiplicative factor appearing in the individual rationality constraint:
\[
\Gamma^B(\lambda)-\Gamma^C(\lambda) = \frac{(1-\lambda)m+(\lambda^2-\lambda)M}{1+\lambda}. 
\]



\subsection{Optimal screening mechanisms and optimal timeline choice}\label{sec:screeningtimelineopt}

We consider the timelines one after the other.\footnote{Establishing the existence of an optimal mechanism in timelines A and B uses classical methods of pointwise maximization whereas the problem in timeline C can be rewritten into a calculus of variations problem with constraints for which we show existence of a solution and give a recipe in the online appendix on how to find it in many typical examples.} Finally, for the case that the monopolist can pick the timeline we establish general results about the optimality of the timeline.

\paragraph{Timeline A.}
We establish in the appendix that the profit function for an incentive compatible and individually rational menu of contracts for timeline A looks as follows.
\begin{equation}\label{eq:profituncertaintyA}
\Pi^{A} = \int_{1}^{\bar\lambda}\left[ \Psi^A(s)v(q(s))-cq(s)\right]G(ds),
\end{equation}
with 
\begin{equation*}
\label{eq:psia}
\Psi^A(s) = \Gamma^A(s)-2m\frac{G(s)}{(1+s)^2g(s)}.
\end{equation*}

$\Psi^A(s)$ is the \emph{virtual valuation} for timeline A. The virtual type $\Gamma^A(s)$ is corrected for the informational rent of the agent represented here by $2m\frac{G(s)}{(1+s)^2g(s)}$.

The monopolist problem is thus maximizing \eqref{eq:profituncertaintyA} under the constraint that $q$ be non-increasing and that individual rationality is fulfilled. 


We give a specific example of the solution for timeline A.

\paragraph{\textbf{Example 1.}} Assume 
\[
G = uniform([1,2]). 
\]

We take $F=uniform([0,1])$ i.e. the intrinsic valuations are uniformly distributed over the interval $[0,1]$. Finally, intrinsic utility is given by $v(q) = \sqrt{q}$. The virtual valuation is calculated to be

\[
\Psi^A(s) = \frac{2}{(1+s)^2},\quad s\in [1,2].\footnote{Note that this is always positive and decreasing. There is no need to exclude types of high loss aversion from the mechanism or to use ironing techniques.} 
\]

The optimal allocation rule is
\begin{equation*}
q^A(\lambda) = \frac{\Psi^A(\lambda)^2}{4c^2}.
\end{equation*}
In particular no type is excluded from the mechanism.
One calculates that optimal profit for timeline A is $R^A(c) =\frac{65}{64\cdot81\cdot c}$. 

\paragraph{Timeline B.}
The profit function for an incentive compatible and individually rational mechanism for timeline B looks as follows.
\begin{equation}\label{eq:profituncertaintyB}
\Pi^{B} = \int_{1}^{\lambda_0}\left[ \Psi^B(s)v(q(s))-cq(s)\right]G(ds),
\end{equation}
with 
\begin{equation*}
\label{eq:psib}
\Psi^B(s) = \Gamma^B(s)-2(m+M)\frac{G(s)}{(1+s)^2g(s)}.
\end{equation*}

$\Psi^B$ is the \emph{virtual valuation} for timeline B. Note that it is weakly lower than the virtual valuation for timeline A. This is because as noted in Proposition \ref{thm:oneagenthelp} the realization effect is negative in expectation due to loss aversion. In comparison to timeline A, this results in a lower virtual type due to lower informational rents. 

The problem of the monopolist is maximizing \eqref{eq:profituncertaintyB} under the constraint that $q$ is non-increasing and that individual rationality is fulfilled.

\paragraph{Timeline C.}

We show in the appendix that the objective function of the designer
can be written as 

\begin{equation*}
\Pi^{C} = f\cdot G(\hat\lambda) + \int_1^{\hat\lambda} \left\lbrace\left[\Gamma^C(\lambda)+ M\frac{G(\hat\lambda)-G(\lambda)}{g(\lambda)}\right]v(q(\lambda))-cq(\lambda)\right\rbrace dG(\lambda)  
\end{equation*}
for a threshold type $\hat\lambda$ so that $(q(\lambda),t(\lambda)) = (0,0)$ whenever $\lambda\ge \hat\lambda$ (exclusion). 

Thus the problem of the designer for timeline C is

\begin{align}
\begin{split}\label{eq:R-Cscreening2}
\max_{\hat\lambda\in[1,\bar\lambda],f\in\R,q(\cdot)}& \int_1^{\hat\lambda} \left\lbrace\left[\Gamma^C(\lambda)+ M\frac{G(\hat\lambda)-G(\lambda)}{g(\lambda)}+f\right]v(q(\lambda))-cq(\lambda)\right\rbrace dG(\lambda)\\&
\text{s.t. }\\& (1)\text{ } q(\cdot)\text{ is non-increasing,}\\& (2)\text{ }
\left[\Gamma^B(\lambda)-\Gamma^C(\lambda)\right]v(q(\lambda))-M\int_1^{\lambda}v(q(s))ds\geq f,\quad \lambda\in [1,\hat\lambda].
\end{split}
\end{align}

Here $\Psi^C(s,f) = \Gamma^C(s)+ M\frac{G(\hat\lambda)-G(s)}{g(s)}+f$ for $s\le \hat\lambda$ is the \emph{virtual valuation} for timeline B. The virtual type $\Gamma^C$ is again corrected for the information rent and the (net) lump-sum subsidy $f$ to the participation self. Condition (2) in the above program is just a rewriting of the individual rationality constraint. $f$ corresponds to a type-independent subsidy which may need to be paid to ensure individual rationality. This subsidy accounts for the externality which the self who chooses the bundle exerts on the participation self. 

We show that a solution for timeline C always exists. In the online appendix we also show how to characterize it completely under some regularity requirements which correspond to a \emph{no ironing} condition in our setting. The second part of the following general result has a straightforward proof though.

\begin{proposition}\label{thm:negativef}
1) There always exists an optimal mechanism for timeline C.

2) If $F$ has support in the non-negative numbers the optimal fixed payment $f$ in an optimal mechanism for timeline C is negative, whenever profits are positive. 
\end{proposition}

2) implies that in many cases the optimal payment schedule $t(\cdot)$ in timeline C is discontinuous in $\lambda$: high loss aversion types are excluded from the mechanism whereas all types who are served may receive a fixed, type-independent transfer from the monopolist to the agents, whenever the monopolist sells some positive amount. Due to the discrepancy between the self choosing the contract and the self deciding whether to participate the former self exerts an externality on the latter by disregarding the surprise effect. This externality can be partially alleviated without adversely affecting incentive compatibility in the second period by optimally transferring a fixed amount $f<0$ to the self at time zero. 

\paragraph{Optimal timeline.}

Assuming that the monopolist can pick the timeline the following Theorem is the main result of this subsection.

\begin{theorem}
\label{thm:AbeatsBscreening}
For the screening model timeline A is weakly better than timeline B which is weakly better than timeline C. 
\end{theorem}
That timeline A is better than B is a direct consequence of the fact that the virtual valuation for timeline B given by $\Psi^B$ in \eqref{eq:psib} is strictly lower than the virtual valuation for timeline A given by $\Psi^A$. Timeline A and B share the surprise effects at the participation decision moment but the realization effect in the good dimension is absent from timeline B. This results in an increased willingness to pay for every type when compared to timeline B.

Intuitively, timeline $C$ may have a more favorable incentive compatibility situation overall than timeline $A$ since the payments from the buyers are not scaled down by $1+\lambda$. The latter happens in timeline A because of the negative surprise effect in the money dimension. On the other hand, as Proposition \ref{thm:negativef} shows, whenever employing timeline C the monopolist has to subsidize participation with a lump-sum payment independent of types. As it turns out for our specification of the problem this subsidy is too costly even for small marginal costs $c$.



The optimality of timeline A relies on two assumptions. First, our Assumption (S) corresponds to assuming a relatively `small' $\lambda^m\mu^m$ in subsection \ref{sec:prefs}. This makes for a small negative surprise effect in the money dimension in timeline A. Second, we have assumed that the monopolist can not offer stochastic payment schedules $t(\lambda)\in \Delta(\R)$. If the latter was possible then the incentive compatibility situation in timeline C is on one hand better than with deterministic contracts since the monopolist can use type-dependent lotteries as an additional screening device and on the other hand worse as now for every type \emph{ceteris paribus} the perceived payments are higher. We conjecture that relaxing these two assumptions may result in timeline C optimality in some cases whereas timeline A still dominates timeline B.\footnote{Details that timeline A still dominates B when relaxing the two assumptions are available upon request.} 


\section{Private information about intrinsic type}\label{sec:multiagent}


In this section we look at multi-agent mechanism design under the assumption that private information concerns the intrinsic type of other players. Behavioral parameters of news utility are common knowledge. For simplicity of exposition we focus on agents whose intrinsic type spaces are identical. The results about incentive compatibility and individual rationality can be generalized to asymmetric agents without difficulty.\footnote{See online appendix for applications to optimal mechanisms with asymmetric agents and timeline A in two cases: optimal auctions with asymmetric agents and bilateral trade.}



\subsection{Set Up}\label{sec:multiagentsetup}


We consider a group of agents $i=1,\dots,N$ who can potentially take part in a mechanism. We assume throughout the designer has standard Expected Utility risk neutral preferences, is interested in revenue maximization and that she has full commitment. 

An agent $i$ derives intrinsic utility from consumption of a profile $a$ of consumption goods coming from a set of allocations $\mathcal{A}$ as well as from a (net) monetary transfer to the principal which is denoted by $t_i$. Formally here $\mathcal{A}$ is assumed to be a compact, connected subset with non-empty interior of $\R^N$. We assume that intrinsic utility for each agent $i$ with type $\theta_i$ is quasilinear and of the form \eqref{eq:intrinsic}. Types are one-dimensional and are given by the interval $\Theta=[\underline{\theta},\bar{\theta}]$ consisting of non-negative numbers. $\theta_i$ is agent $i$'s (intrinsic) type and we denote $\Theta=\Theta_1\times\dots\times\Theta_n$ the product of the type spaces. 

If the agent doesn't participate in the mechanism the value of her intrinsic utility in the allocation dimension is given by a number denoted $v_i(\emptyset)$. We require throughout that $v_i(\emptyset)$ is either the \emph{maximal} or the \emph{minimal} value that $v_i$ can take, i.e. ($v_i(\emptyset)\in \{\sup v_i,\inf v_i \}$).\footnote{All of our applications both in the main paper as well as in the online appendix fulfill this assumption.} Moreover, in the money dimension we assume the agent doesn't expect any transfers in the absence of the mechanism.\footnote{The online appendix comments on the case of non-trivial outside options in the money dimension. Incentive compatibility and individual rationality characterizations are similar to the ones in this section. We focus here on the trivial case for ease of exposition.}

For the model presented in this section we make the following informational assumption.

  \paragraph{\textbf{Assumption (A):}} The private information of the agent $i$ consists of $\theta_i$ (her \emph{intrinsic} type). The agents know their type at the outset of any interaction with the designer (interim stage). All agents and the designer have a common knowledge prior for the type profile $(\theta_1,\dots,\theta_N)$. The types across agents are i.i.d. and the common marginal distribution of $\theta_i$, denoted $F$, has a continuously differentiable, strictly positive density $f:[\underline{\theta},\bar\theta]\ra\R_+$.
\vspace{3mm}\\

\begin{figure}[H]
  \centering
    \includegraphics[width=0.65\textwidth]{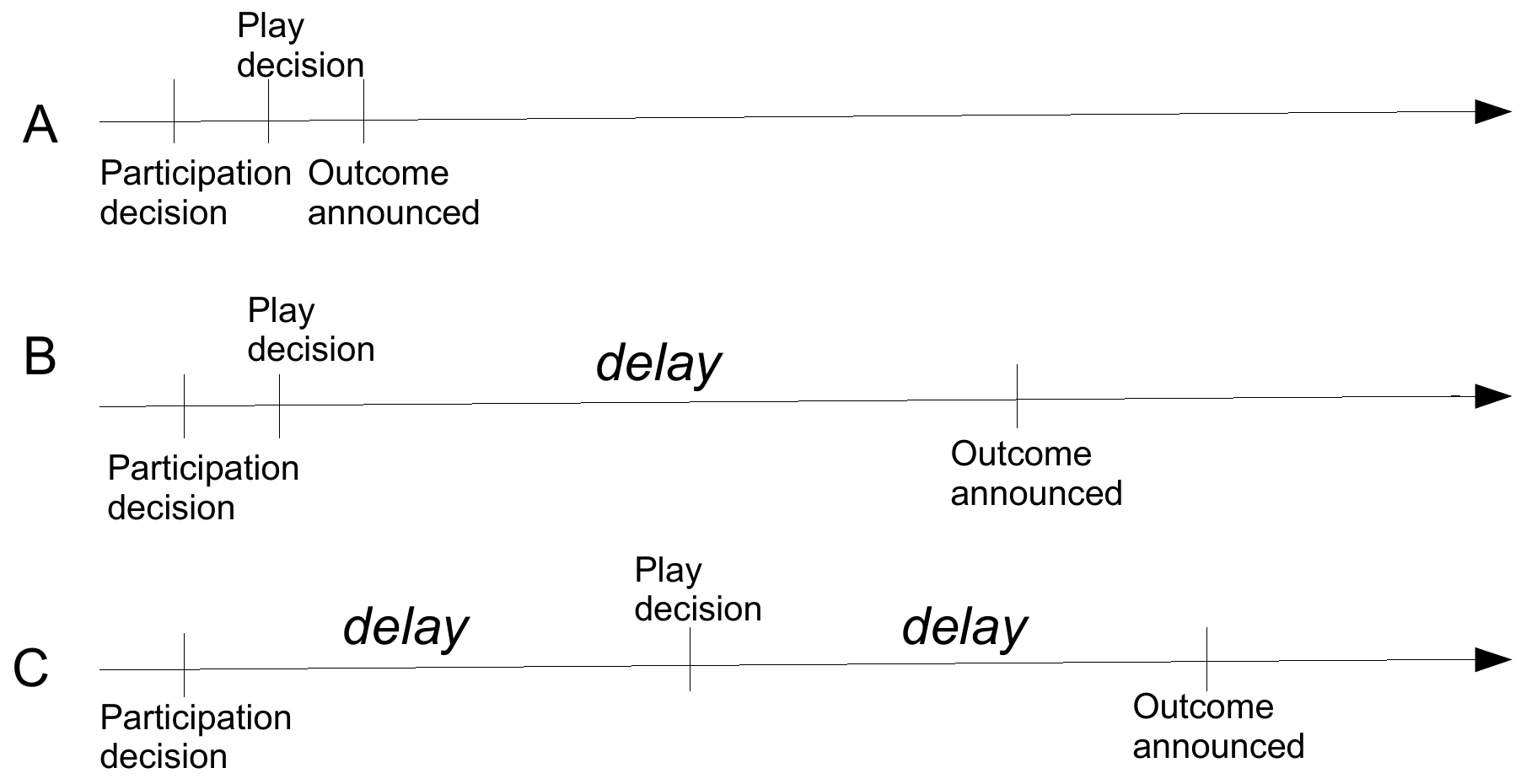}
      \caption{Timelines for the multiple agent model with uncertainty about intrinsic types.}\label{fig:tl}
\end{figure}

The timeline of the mechanism is either fixed due to technological constraints or a choice variable of the designer. In the latter case he declares at the beginning the timeline he commits to. As established in Proposition \ref{thm:oneagenthelp} the relevant timelines for the analysis are in Figure \ref{fig:tl}.

For any fixed timeline the designer can in principle consider arbitrarily complicated mechanisms. We restrict the analysis without loss of generality to \emph{direct} mechanisms. For a given timeline, a direct mechanism asks the agents to report their private information and assigns as a function of their reports an allocation $a$ from $\mathcal{A}$ and (net) transfers $t_i$ to the designer. The restriction to direct mechanisms is justified by the \emph{revelation principle}.\footnote{The revelation principle holds true for all models we consider in this section. We establish this fact in the Online Appendix. 
} 

Formally, a direct mechanism is a mapping which gives for each type report profile an allocation in $\A$ as well as payments from the agents to the designer \emph{together} with a timeline $T\in \{A,B,C\}$. Formally, for any $T\in \{A,B,C\}$ a direct mechanism is a map as follows.

\begin{equation}\label{eq:directM}
\mathcal{M}^T=(q,t_1,t_2,\dots,t_n):\Theta_1\times\dots\times\Theta_n\ra \mathcal{A}\times\R^n. 
\end{equation}

The uncertainty each agent faces in a given mechanism $\mathcal{M}^T$ derives only from not knowing the other agents' types. In terms of subsection \ref{sec:prefs} we are considering agents who are offered menus $M_i$ of lotteries over $\Delta(\A\times\Theta_{-i}\times\R)$ parametrized by $\theta_{i}\in [\underline{\theta},\bar{\theta}]$ and who have to time their decisions according to the timeline $T$. Given the general form of mechanisms allowed and the full support assumption on $F$ it is without loss of generality to assume that the mechanisms are not randomized, i.e. that the designer doesn't have a randomization device at her disposal.\footnote{She can use the type draws $\theta$ to induce desired distributions on payments.}



\subsection{Incentive Compatibility and Individual Rationality}\label{sec:multiagenticir}

In the following we take as given a direct mechanism $\mathcal{M}^T$ as in \eqref{eq:directM}. 

Fix an agent $i\in\{1,\dots,N\}$. For any distribution $G\in \Delta(\A\times\R)$ giving the distribution of pairs $(a,t_i)$ induced from the play under the mechanism over $\A\times\R$ the marginal of $G$ over $\A$ is denoted by $G^a$ and over $t_i$ is denoted by $G^t$.
We call the following term the \emph{news utility of agent $i$ from changing beliefs from $H$ to $G$ when the type of the agent is $\theta_i$.} 
\begin{equation}\label{eq:newsutilitygeneral}
\mathcal{N}_i(G|H|\theta_i) = \mu^g_i\int_{0}^1 \xi^g_i\left(v_i(c_{G^a}(p))\theta_i-v_i(c_{H^a}(p))\theta_i\right)dp+\mu^m_i\int_{0}^1\xi^m_i\left(c_{H^t}(p)-c_{G^t}(p)\right)dp.
\end{equation}
Here $\mu^g_i, \mu^m_i, \lambda_i^g,\lambda_i^m$ are the agent-specific behavioral parameters for news utility as in subsection \ref{sec:prefs}.  For future reference we also define the aggregate news utility parameters $\Lambda_i^g = \mu_i^g(\lambda_i^g-1)$ and $\Lambda_i^m = \mu_i^m(\lambda_i^m-1)$.


Assume that in the mechanism play other agents $-i$ decide to participate and report their types truthfully to the designer, whereas agent $i$ of type $\theta_i$ decides to report $\hat \theta_i$ upon a positive participation decision. Define $V_i(\hat\theta)$ as the expected value of $v_i$ and $T_i(\hat\theta_i)$ the expected value of $t_i$ under these reporting strategies from the perspective of agent $i$.\footnote{Formally, $V_i(\hat\theta_i) = \E_{\theta_{-i}}[q(\hat\theta_i,\theta_{-i})]$ and $T_i(\hat\theta_i) = \E_{\theta_{-i}}[t(\hat\theta_i,\theta_{-i})]$.} Finally, define $T^+_i(\hat\theta_i) = \E_{\theta_{-i}}[\max\{t_i(\hat\theta_i,\theta_{-i}),0\}]$, the the expected transfer of type $\hat\theta_i$ from agent $i$ to the designer. 

The news utility from the realization effect if the agent decides to participate for the case $v_i(\emptyset) =\inf_{a\in\mathcal{A}}v_i(a)$ is 

\[
\mathcal{N}_{i}(\hat{\theta}_i|\emptyset|\theta_i)=\mu^g_iV_i(\hat{\theta}_i)\theta_i-\mu^g_iv(\emptyset)\theta_i - \mu_i^mT_i(\hat{\theta}) - \Lambda_i^mT_i^+(\hat{\theta}),
\]
whereas for the case $v_i(\emptyset) =\sup_{a\in\mathcal{A}}v_i(a)$ it is 

\[
\mathcal{N}_{i}(\hat{\theta}_i|\emptyset|\theta_i)=\lambda_i^g\mu^g_iV_i(\hat{\theta}_i)\theta_i-\lambda_i^g\mu^g_iv(\emptyset)\theta_i - \mu_i^mT_i(\hat{\theta}) - \Lambda_i^mT_i^+(\hat{\theta}).
\]
In contrast to the first case of $v_i(\emptyset) =\inf_{a\in\mathcal{A}}v_i(a)$, in the second case of $v_i(\emptyset) =\sup_{a\in\mathcal{A}}v_i(a)$ the utility difference between pre-mechanism belief and post-participation decision is weighted additionally by $\lambda_i^g$. This is because of loss aversion.

Consider now the \emph{news utility from the realization of the outcome} of the mechanism for timelines B and C. If agent $i$ of type $\theta_i$ has reported $\hat\theta_i$ and the realized part of the outcome of the mechanism relevant to agent $i$ is 
\[
\left(q(\hat\theta_i,\hat\theta_{-i}),t_i(\hat\theta_i,\hat\theta_{-i})\right),
\]
in addition to intrinsic utility she experiences news utility in the good dimension of
\begin{align}\label{eq:realizationnugood}
\begin{split}
&\mu_i^g\int_{\theta_{-i}:v_i(q(\hat\theta_i,\hat\theta_{-i}))>v_i(q(\hat\theta_i,\theta_{-i}))}(v_i(q(\hat\theta_i,\hat\theta_{-i}))-v_i(q(\hat\theta_i,\theta_{-i})))dF_{-i}(\theta_{-i}) \\& + \mu_i^g\lambda_i^g \int_{\theta_{-i}:v_i(q(\hat\theta_i,\hat\theta_{-i}))<v_i(q(\hat\theta_i,\theta_{-i}))}(v_i(q(\hat\theta_i,\hat\theta_{-i}))-v_i(q(\hat\theta_i,\theta_{-i})))dF_{-i}(\theta_{-i})
\end{split}
\end{align}
Again, the second summand which is non-positive, is weighted by $\lambda_i^g$ due to the loss aversion in the good dimension. 

The agent experiences news utility in the money dimension given by 

\begin{align}\label{eq:realizationnumoney}
\begin{split}
&\mu_i^g\int_{\theta_{-i}:t_i(\hat\theta_i,\hat\theta_{-i})<t_i(\hat\theta_i,\theta_{-i})}(t_i(\hat\theta_i,\hat\theta_{-i})-t_i(\hat\theta_i,\theta_{-i}))dF_{-i}(\theta_{-i}) \\& - \mu_i^g\lambda_i^g \int_{\theta_{-i}:t_i(\hat\theta_i,\hat\theta_{-i})>t_i(\hat\theta_i,\theta_{-i})}(t_i(\hat\theta_i,\hat\theta_{-i}))-t_i(\hat\theta_i,\theta_{-i}))dF_{-i}(\theta_{-i}).
\end{split}
\end{align}
The news utility experienced at the realization moment of the mechanism, denoted $\mathcal{N}_i((\hat\theta_{i},\hat\theta_{-i})|\hat\theta_{i}|\theta_i)$, is thus the sum of \eqref{eq:realizationnugood} and 
\eqref{eq:realizationnumoney}. 

For timelines B and C at the report decision moment the agent takes into account the news utility from the realization effect in expectation. We show that agent $i$ of type $\theta_i$ has an overall term of \emph{expected realization news utility} of

\begin{equation}\label{eq:helpdegen}
-\Lambda_i^g \Gamma_i^g(\hat\theta_i)\theta_i-\Lambda_i^m\omega_i(\hat\theta_i),
\end{equation}
where $\Gamma_i^g(\hat\theta_i)$ and $\omega_i(\hat\theta_i)$ are non-negative. These `frictions' are zero if and only if respectively the allocation $q$ and the transfer $t_i$ don't depend on the realization of the types of the other players, that is don't depend on $\theta_{-i}$. The realization effect thus lowers the decision utility of an agent at the reporting stage. 
Note that in the screening model of section \ref{sec:screening} we only had one such expected news utility term $M$ which was exogenously given. Here the expected news utility terms $\Gamma_i^g(\hat\theta_i)$ and $\omega_i(\hat\theta_i)$ are influenced by the designer through the mechanism choice as well as the timeline choice, if the latter is a choice variable. This additional flexibility has important implications for revenue maximization as we will see. 


The following Proposition registers the decision utilities for agents in each of the participation and reporting stages.
\begin{proposition}
\label{thm:IC}
For the timelines in Figure \ref{fig:tl}, the decision utilities at the type reporting stage have a quasilinear, product form. Namely, the utility of agent $i$ of type $\theta_i$ from reporting type $\hat\theta_i$ is of the form
\begin{equation}\label{eq:product}
\mathcal{V}^t_i(\hat{\theta}_i|\theta_i) = \mathcal{W}_i^t(\hat{\theta}_i)\theta_i- \Upsilon^t_i(\hat{\theta_i}), \quad t=A,B,C.
\end{equation}
Depending on the timeline $t$ the functions $\mathcal{W}_i^t, \Upsilon^t_i:[\underline{\theta},\bar\theta]\ra\R$ have the following form in the case $v_i(\emptyset) = \inf v_i$
\begin{itemize}
\item Timeline A: $$\mathcal{W}_i^A(\theta_i)=(1+\mu^g_i)V_i(\theta_i)-\mu^g_iv(\emptyset),\quad \Upsilon^A_i(\theta_i) = (1+\mu_i^m)T_i(\theta_i) + \Lambda_i^mT_i^+(\theta_i).$$
\item Timeline B: 
$$\mathcal{W}_i^B(\theta_i)=(1+\mu^g_i)V_i(\theta_i)-\mu^g_iv(\emptyset)-\Lambda_i^g\Gamma_i^g(\theta_i),\quad \Upsilon^B_i(\theta_i) = (1+\mu_i^m)T_i(\theta_i) + \Lambda_i^m(T_i^+(\theta_i)+\omega_i(\theta_i)).$$
\item Timeline C: 
$$\mathcal{W}_i^{C}(\theta_i)=V_i(\theta_i)-\Lambda_i^g\Gamma_i^g(\theta_i),\quad \Upsilon^{C}_i(\theta_i) = T_i(\theta_i) + \Lambda_i^m\omega_i(\theta_i).$$
\end{itemize}
For the case $v_i(\emptyset) = \sup v_i$ the only change is in timelines $t=A,B$ where the terms $\mathcal{W}_i^t$ change into 

\begin{itemize}
\item Timeline A: $$\mathcal{W}_i^A(\theta_i)=(1+\lambda_i^g\mu^g_i)V_i(\theta_i)-\lambda_i^g\mu^g_iv(\emptyset),$$
\item Timeline B: 
$$\mathcal{W}_i^B(\theta_i)=(1+\lambda_i^g\mu^g_i)V_i(\theta_i)-\lambda_i^g\mu^g_iv(\emptyset)-\Lambda_i^g\Gamma_i^g(\theta_i).$$
\end{itemize}
\end{proposition}

Note that \eqref{eq:product} is reminiscent of the classical utility assumption \eqref{eq:intrinsic}. We name the terms $\mathcal{W}_i^t$ \emph{perceived valuations}. In the equilibrium of the mechanism they give the marginal expected valuation of the allocation for the agent after taking into account news utility effects, be it the surprise effect (timelines A,B), the realization effect (timelines B,C) or both (timeline B). When news utility is absent, all perceived valuations $\mathcal{W}_i^t$ are equal to the expected intrinsic valuation $V_i$. 

We name the terms $\Upsilon_i^t$ the \emph{perceived transfers}. In case news utility is absent, these terms become equal to the expected interim transfers $T_i$ for all timelines. In the equilibrium of the mechanism they give the effect of money transfers on the decision utility in the type reporting stage after taking into account news utility effects. 

With these definitions, classical results help give a full characterization of incentive compatibility for all timelines.

\begin{proposition}[Incentive Compatibility]\label{thm:icgeneral}
A direct mechanism $\mathcal{M}$ is incentive compatible for the timeline $T\in \{A,B,C\}$ if and only if the perceived valuations $\mathcal{W}^T_i$ are non-decreasing.
\\\indent If this is the case, we have the following Mirrlees representation for the interim utility at the reporting stage $\mathcal{V}^T_i(\theta_i)= \mathcal{V}^T_i(\theta_i|\theta_i)$
\begin{equation*}\label{eq:mirrleesgeneral}
\mathcal{V}^T_i(\theta_i) = \mathcal{V}^T_i(\underline{\theta}_i|\underline{\theta}_i) +  \int_{\underline{\theta}_i}^{\theta_i} \mathcal{W}^T_i(s)ds.
\end{equation*}
It also holds 
\begin{equation*}\label{eq:upsilongeneral}
\Upsilon_i^T(\theta_i) = \mathcal{W}^T_i(\theta_i)\theta_i - \mathcal{V}^T_i(\underline{\theta}_i) - \int_{\underline{\theta}_i}^{\theta_i} \mathcal{W}^T_i(s)ds.
\end{equation*}
\end{proposition}

Analogously to the classical case, incentive compatibility is equivalent to a monotonicity condition. Due to news utility effects this monotonicity condition is applied to \emph{perceived} valuations instead of expected intrinsic valuations. 

We now turn to the individual rationality requirement for incentive compatible mechanisms. For all timelines each agent experiences the surprise effect of the mechanism and, given sophistication, takes into account her own future behavior when deciding to participate. In particular, when facing an incentive compatible mechanism she knows she is going to reveal her true type to the designer at the reporting stage. Under the assumption that the other agents play truthfully reporting the correct type induces a distribution over good consumption and money transfers. 



We assume the outside option of the mechanism is degenerate: for agent $i$ of type $\theta_i$ the outside option is $v_i(\emptyset)\theta_i$. The mechanism is individually rational for an agent $i$ when the equilibrium utility it offers in period one is at least as high as $v_i(\emptyset)\theta_i$ for all types $\theta_i\in [\underline{\theta},\bar\theta]$. This equilibrium utility incorporates not only the expected consumption and transfers from the realization of the mechanism, but also news utility in the form of the surprise effect for all timelines and in the form of the expected realization effect for timelines B,C. The following Proposition summarizes this in participation utility formulas.

\begin{proposition}\label{thm:IRgeneral}(Individual Rationality)
An incentive compatible mechanism $\mathcal{M}$ is individually rational for the respective timelines if the following is fulfilled.
\begin{itemize}
\item Timeline A: $\quad \mathcal{W}_i^A(\theta_i)\theta_i- \Upsilon^A_i(\theta_i)\geq v_i(\emptyset)\theta_i$, for all $\theta_i\in[\underline{\theta},\bar\theta]$. 
\item Timeline B: $\quad \mathcal{W}_i^B(\theta_i)\theta_i- \Upsilon^B_i(\theta_i)\geq v_i(\emptyset)\theta_i$, for all $\theta_i\in[\underline{\theta},\bar\theta]$.
\item Timeline C: $\quad \mathcal{W}_i^B(\theta_i)\theta_i- \Upsilon^B_i(\theta_i)\geq v_i(\emptyset)\theta_i$, for all $\theta_i\in[\underline{\theta},\bar\theta]$.
\end{itemize}
\end{proposition}
Individual rationality for timelines A and B is equivalent to the requirement that the equilibrium \emph{decision} utility in the \emph{reporting} stage exceeds the value of the outside option $v_i(\emptyset)\theta_i$. This is because in both of these timelines there is no delay between the participation decision and the reporting decision. In contrast, in timeline C the reporting decision doesn't take into account the bygone \emph{surprise effect}. The participation self has to take both news utility effects into account. Therefore the participation decision utilities are the same as for timeline B, except that the terms $V_i,T_i,T_i^+$ are now determined in the reporting stage. 


\subsubsection{Ex-post Efficiency in public good provision}\label{sec:multiagentexpostpublic}

In this subsection we illustrate how loss aversion affects incentive compatibility results from classical Bayesian mechanism design. Namely, we show how ex-post efficiency in a symmetric public good provision setting may fail to be incentive compatible, even though it never does so in the absence of news utility effects. We also illustrate how this issue may be overcome whenever there are enough players in the game. 

Assume a society of $N\geq 2$ identical agents whose intrinsic utility of a public good is private information and is independently and identically distributed according to $F$ with support $[\underline{\theta},\overline{\theta}]$ with $\underline{\theta}\geq 0$.\footnote{The last requirement means the public good would always be (weakly) desirable if it is implemented without transfers.} Assume the agents have quasilinear intrinsic utility and denote by $q$ the probability of provision of the public good: the intrinsic valuation has the form $v(q)=q$. Let $c(N)$ be the average costs so that providing the public good will cost to the social planner $Nc(N)$. 

Maximizing welfare under \emph{complete information} and under the assumption that the social planner has the funds for provision gives the ex-post efficiency rule:  

\begin{equation}
\label{eq:expostefficiencypublicdelay}
q(\theta_1,\dots,\theta_n) = \begin{cases} 1 &\mbox{if } (1+\mu^g)\sum_{i=1}^n\theta_i\geq Nc(N)  \\ 
0 & \mbox{otherwise. } \end{cases}
\end{equation}

This is the first-best rule under news utility. Denote by $\tilde{c}(N)=\frac{c(N)}{1+\mu^g}$, the normalized per-person provision costs. 
If an agent is of type $\theta$ her interim probability of public good provision under \eqref{eq:expostefficiencypublicdelay} is $Q(\theta)=1-F^{*(N-1)}(N\tilde{c}(N)-\theta)$.\footnote{Here with $F^{*(N)}$ we denote the $N$-times convolution of a distribution $F$, i.e. the distribution of the sum of $N$ i.i.d. draws of $F$. }

In the classical setting without news utility and loss aversion ex-post efficiency is always incentive compatible. This is not always the case in the presence of news utility.\footnote{Ex-post efficiency is always incentive compatible with news utility if $c(N)\leq (1+\mu^g)\underline{\theta}$ or if $c(N)\geq (1+\mu^g)\bar{\theta}$. We omit these uninteresting cases and focus on the interesting case $(1+\mu^g)\underline{\theta}< c(N)<(1+\mu^g)\bar{\theta}$ in the following.}

\begin{proposition}\label{thm:prop1public} Assume $(1+\mu^g)\underline{\theta}< c(N)<(1+\mu^g)\bar{\theta}$.

1) Ex-post efficiency is always incentive compatible for timeline A. 

If $\Lambda^g\leq 1+\mu^g$ then ex-post efficiency is incentive compatible for timeline B. 

If $\Lambda^g\leq 1$ then ex-post efficiency is incentive compatible for timeline C.

2) If $\Lambda^g> 1+\mu^g$, it can happen that $q$ is not incentive compatible for timeline B, for example if $F$ is concentrated in the vicinity of $\underline{\theta}$ and $Nc(N)$ is big enough. This impossibility becomes more common with higher loss aversion in the consumption dimension. 

The same kind of result holds true for timeline C, if $\Lambda^g>1$.

3) Let $\mathbb{E}[F]$ be the mathematical expectation of the distribution $F$. \newline If $\underline{\theta}<\limsup_{n\rightarrow \infty}\tilde{c}(N)< \mathbb{E}[F]$, then $q$ is always incentive compatible whenever $N$ is high enough with \emph{every} timeline. 
\end{proposition}

Note that in timeline A the model is isomorphic to a classical quasi-linear model where the type spaces are $[(1+\mu^g)\underline{\theta},(1+\mu^g)\bar\theta]$ and the transfers of each agent are scaled down by $\frac{1}{1+\lambda^m\mu^m}$. The same argumenst as in the classical setting deliver all of 1). 

The result of part 2) of Proposition \ref{thm:prop1public} for timelines B, C is an instance of equilibrium non-existence. It is similar to the result in Theorem 1 of \cite{dato} who show lack of existence of equilibrium in the related CPE model.\footnote{CPE is defined in \cite{kr2} and has become popular in the applied behavioral literature. \cite{mr} show that when CPE is interpreted as a static risk preference it leads to choice over lotteries which may violate monotonicity with respect to first-order-stochastic-dominance (FOSD), whenever the loss aversion parameter $\lambda^m$ is high enough. Part 2) of the above Proposition is a result in the same spirit.}  

When $N$ is very large the comparison between average intrinsic valuation $\E[F]$ and average costs scaled by $1+\mu^g$ decides on provision. Due to the law of large numbers then, if the average costs of provision fulfill the inequality in part 3) of the Proposition, the good is provided with high ex-ante probability so that the expected news utility terms from the realization effect weigh less on the decision of the agent. This effect helps preserve incentive compatibility even under relatively high loss aversion.

\subsection{Optimal Symmetric Auctions}\label{sec:multiagenttimelineopt}


In this subsection we focus on symmetric unit auctions.\footnote{The online appendix contains further results: optimal unit auctions for asymmetric agents in timeline A as well as optimal symmetric auctions for timeline B. In both cases we show how news utility changes results and intuitions from the classical setting.} We first define the environment.

Let $\Delta = \{(q_1,q_2,\dots,q_N)\in \R_+^N: \sum_{i=1}^n q_i\leq 1\}$ be the feasible allocations of a single unit of a good to be auctioned off between $N$ bidders.  

A direct mechanism $\mathcal{M}$ for the auction with a fixed timeline $T\in \{A,B,C\}$ is a mapping $$(q_1,\dots,q_N,t_1,\dots,t_N):\Theta\ra \Delta\times\R^n$$ giving as a function of reports for each agent the probability that she gets the good and the payment to the auctioneer.

We make the classical assumption $v_i(q)=q_i$, i.e. the intrinsic value of the good of an agent $i$ is equal to the probability that the good ends up with agent $i$.

To ensure that there are no incentive compatibility issues as in subsection \ref{sec:multiagentexpostpublic} we add a parametric restriction for the news utility parameters to the classical regularity assumption for the virtual valuation of the agents. 
\vspace{2mm}

\textbf{Assumption (A2)} 
\begin{itemize}
\item No Dominance of News Utility in the good dimension:
$$\Lambda^g = \mu^g(\lambda^g-1)\le 1.$$
\item Regularity: The virtual valuation of $F$ given by the function \\$\gamma(t):[\underline{\theta},\bar{\theta}]\ra \R,\gamma(t) = t-\frac{1-F(t)}{f(t)}$ is strictly increasing.
\end{itemize}
\vspace{2mm}
The restriction for the aggregate news utility parameter ensures that no bidder shows preference for a stochastically dominated allocation in the good dimension in timelines A and B.\footnote{Recall the related discussion in \cite{mr}. This parameter restriction also appears in other settings in the applied behavioral literature, such as \cite{herwegetal}.} 

Regularity corresponds to the classical assumption first introduced in \cite{myerson}. It is fulfilled for many natural examples and allows for simple characterizations of the optimal auction.\footnote{For timelines A and B one could use the methods from \cite{toikka} whenever $\Lambda^g>1$. The model here is separable according to the terminology in \cite{toikka} (see section 3 of his paper).} 

\begin{proposition}\label{thm:ic-auctions}
1) Under Assumption (A2) the optimal allocation rule for timelines $A,B$ is the same as the Myersonian rule, that is, given $\theta^*$ such that $\gamma(\theta^*) = 0$ sell to any of agents with the highest type $\theta^m=\max_{i\le n}\theta_i$ as long as $\theta^m\ge \theta^*$.

2) Optimal auctions in timeline B are all-pay. That is, it is optimal to fully insure the bidder against the uncertainty she is facing in the transfers.

3) Optimal auctions in timeline C are usually not all-pay. That is, it may be optimal to not fully insure a positive measure of types against the uncertainty they are facing in the money dimension. Moreover, the optimal threshold type is usually different from the one in timelines A and B. 
\end{proposition}

An indirect optimal auction is an auction with a reservation price and which follows the respective timeline. In the case of timeline B it is additionally all-pay: each bidder has to pay her bid. Bidders who would never bid above the reservation price bid zero and pay zero.

Part 1) in the case of timeline A follows closely the standard classical proof for symmetric unit auctions.\footnote{See e.g. chapter 3 of \cite{boergers}.} The case of timeline $B$ follows immediately from (the more general) Proposition 6 in the online appendix.

Part 2) follows from the general all-pay result for timeline B which is proven in the online appendix (see Theorem 1 there). A related result is known in the classical literature on Expected Utility agents: optimal revenue maximizing auctions feature degenerate transfers whenever an agent's Bernoulli utility is separable in the consumption and money dimension and the agent is risk averse w.r.t. money (see \cite{maskinriley}).

In timeline C the auctioneer has an additional variable he can use to give incentives  for truth-telling in the reporting stage: the expected news utility terms in the money dimension given by $\omega(\theta)$. Its usage comes at a cost as \emph{ceteris paribus} an $\omega(\theta)>0$ lowers the interim equilibrium utility in the reporting stage for type $\theta$. Thus, when compared to timeline A, timeline C besides the advantage that any expected payments are not shaded down as in timeline A because of the missing surprise effect in the money dimension, it has the additional advantage of having one more choice variable for the auctioneer. Timeline A has the advantage of lacking the negative expected news utility term from the realization effect, which implies that there is no need to subsidize individual rationality as it is usually necessary in timeline C. We show numerically that in many parameter constellations for timeline C the auctioneer decides to make use of $\omega(\theta)$, i.e. leaves some of the types with risk in the money dimension so as to help incentive compatibility in the reporting stage. The next example illustrates this optimal distortion in the money dimension.


\paragraph{Example 3.} Consider an auction with two symmetric agents who satisfy the following assumptions: 
$\lambda^g = 1.2, \mu^g=1, \mu^m\lambda^m = 1$ and distribution of intrinsic type $F= uniform([1,2])$. Note that the lowest intrinsic type has strictly positive utility from getting the good.

Figure \ref{fig:fricC} depicts the optimal distortion in the money dimension for the case of timeline C.

\begin{figure}[ht!]\label{fig:fricC}
  \centering
    \includegraphics[width=0.35\textwidth]{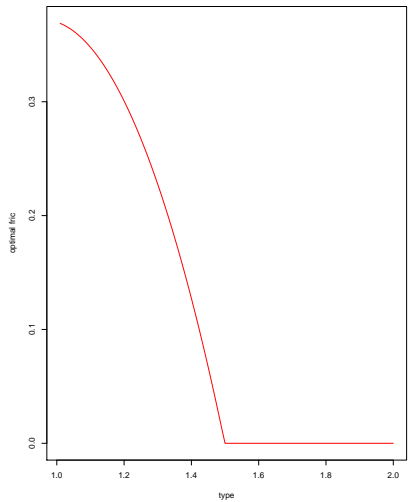}
      \caption{Friction in the money dimension as a function of type for Example 4.}
\end{figure}

Here, optimal fric is the following map as a function of types,

\[
[1,2]\ni\theta\ra c_m\omega(\theta),
\]
where $c_m= \frac{\lambda^m\mu^m}{1+\lambda^m\mu^m}\Lambda^m$ is a normalizing constant depending only on $\lambda^m,\mu^m$. Note that the distortion in the money dimension is decreasing and that there is no distortion at the top.\footnote{The monotonicity in distortion is not a general feature as other numerical exercises show (available upon request). For example, an inverse-U shaped distortion is possible if $F = uniform([0,1])$: auctioneer distorts only intermediate types in the money dimension. The no-distortion at the top property seems to be fulfilled in many numerical examples.}

Another feature of optimal auctions in timeline C is that now the threshold type the auctioneer uses to decide whether to sell the unit at all is different from the classical Myerson one $\gamma(\theta^*)$. The auctioneer has to weigh different effects in timeline C: types have to be subsidized for participating because of the wedge between incentive compatibility and individual rationality and she has to decide which subset of types to fully insure in the money dimension so as to maximize incentives of truthtelling after a positive participation decision. As news utility parameters are varied but so as to keep $\Lambda^g\le 1$ satisfied, the optimal balance of these effects may result for a fixed $F$ fulfilling regularity in both a higher or lower threshold than the classical threshold $\theta^*$ (which satisfies $\gamma(\theta^*) = 0$).

\subsubsection{Timeline Optimality}

In this subsection we assume the timeline is a choice variable of the auctioneer and consider its optimality.

\begin{theorem}
\label{thm:optimaltimelineresults}
1) Timeline A always dominates timeline B in terms of revenue maximization. 

2) There is no uniform ranking of timelines A and C in terms of revenue maximization. 
\end{theorem}

When compared to timeline A, timeline B features the same news utility effects except for an additional negative realization effect coming from the negative expected news utility in the good dimension. This lowers \emph{ceteris paribus} for timeline B the willingness to pay of a bidder so that timeline A is always preferred for revenue maximization.

The second part is proven by example.

\paragraph{Example 4.} We take the same data as Example 4 with the only difference that now we don't fix $\mu^m\lambda^m$ but instead vary it in the interval $[1,2]$. We look at timelines $A, C$. 

\begin{figure}[H]\label{fig:auctionAvsC}
  \centering
    \includegraphics[width=0.28\textwidth]{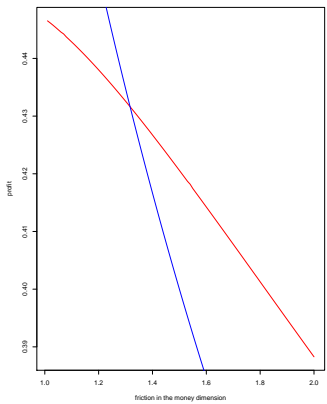}
      \caption{Auction revenues for timelines A (blue) and C (red) for friction in the money dimension in the range $[1,2]$.}
\end{figure}

When $\mu^m\lambda^m$ is low the payments of the agent in timeline A are not shaded as much due to the negative surprise effect in the money dimension. This, and the fact that timeline C features a negative realization effect as well as a possible lump-sum subsidy for participation yield optimality of timeline A whenever $\mu^m\lambda^m$ is low. When $\mu^m\lambda^m$ is high the advantages of timeline C come to bear: there is no shading of payments due to the negative surprise effect in the money dimension and there is an additional (albeit costly) variable which can be used to give incentives in the reporting stage. While it is true that individual rationality may need a lump-sum subsidy, the subsidy is on average small because the lowest possible intrinsic type $(\theta = 1)$ loses a lot in terms of intrinsic utility in case of non-participation, her outside option being zero overall utility. In fact, numerical results available upon request show that the optimal timeline in this example is indeed A if we change the distribution of intrinsic types in Example 4 to $F= uniform([0,1])$.

\section{Conclusions}

In this paper we have considered agents whose preferences are sensitive to changes of beliefs and additionally feature loss aversion. Under the assumption that agents are sophisticated about future behavior we have characterized features of optimal mechanism design in two different models: one where a monopolist screens a single agent according to their loss aversion parameter and one where multiple agents face uncertainty regarding the intrinsic types of the other agents. 
This work can be extended along different directions. 

For simplicity we have assumed that there is no discounting of time. This makes two of the possible timelines equivalent for all purposes. Relaxing that assumption is a fruitful didactic exercise as it would give a more complete picture for the characterization of optimal timeline design.


A major assumption to relax in our model is sophistication. Allowing for naive or partially naive agents in our setting will eventually lead to changes regarding timeline optimality as well as changes in the features of optimal mechanisms. To see how, consider timeline C and assume that the self who decides about participation assumes erroneously that the self who decides about play will stick to her optimal plans. This will imply that the participation-decision self will behave the same in timelines B and C. Knowing this, she won't ask for a lump-sum subsidy in order to participate as was the case under the sophistication assumption. \emph{Ceteris paribus} this lowers the implementation costs for timeline C for the designer. We conjecture that in the case of naive agents timeline C becomes optimal in many more cases than it does in the case of sophisticated agents. 

We haven't considered the case where different agents may be in different timelines or the case where the timeline is not common knowledge for all agents at the start of the game. Moreover, we haven't solved for the optimal timeline in the case of auctions with asymmetric agents. These non-trivial extensions are left for future research.

Finally, a new strand of literature started by papers like \cite{bester} and \cite{maskinmoore} considers designers who don't have full commitment. Relaxing the full commitment assumption in our setting is a very interesting topic left for future research.

\begin{appendices}

\section{Auxiliary Results}\label{sec:a-auxiliary}

\paragraph{A result on expected \emph{future} news utility terms.}
Consider first a one dimensional model. Here we skip the indices $j=g,m$ for simplicity. 

Whenever $G$ is a degenerate distribution corresponding to getting $r$ with probability one, we write for the news utility term comparing the degenerate distribution to  $H$ $\mathcal{N}(r|H)$. It holds

\[
\mathcal{N}(r|H) =\E_{z\sim H}[\xi(r-z)],
\]
where $\xi$ is a piecewise linear gain-loss valuation function as in \eqref{eq:valkt}.
The \emph{expected news utility from the realization of} $H$ is denoted by $-\omega(H)$ and it holds

\[
\omega(H) = \E_{r\sim H}[\mathcal{N}(r|H)].
\]

The following technical Lemma is easy to prove.\footnote{Note that a similar result has been proven in the CPE setting in \cite{eisenhuth}. See Lemma 1 there. Since the proof of our Lemma follows word-for-word his argument we skip it here.}

\begin{lemma}\label{thm:expectedNU} It holds
$$\omega(H) = \Lambda \int \int_{\{z>w\}}(z-w) dH(z)dH(w).$$ 
$\omega(H)$ is nonnegative and equal to zero if and only if $H$ is a degenerate distribution. Moreover, whenever $H$ is supported in the non-negative numbers it holds 

$$\omega(H) \leq \Lambda \E[H].$$
\end{lemma}

We have considered a model, where the uncertainty is one-dimensional. Given the separability assumptions we make in Subsection \ref{sec:prefs} this result implies immediately the proof of part 1) in Proposition \ref{thm:oneagenthelp}. 

\paragraph{Image of news utility terms as a function of distributions.} The following is a characterization for the image of the pair $(\omega(H),\E[H])$ as a function of the one-dimensional distribution $H$. 

\begin{lemma}\label{thm:lmhelpmixed}
For every element $(x,y)\in \R_{+}\times\R$ there exists a binary distribution $H=L(p,b,d) = p\delta_{b}+(1-p)\delta_d$ with $b\geq d$ and $p\in (0,1)$ such that $(\omega(H),\E[H]) = (x,y)$. 
\end{lemma}

\begin{proof}
If $x=0$ then just pick $H= \delta_{y}$. So let's focus on the case $x>0$ which corresponds to $b>d$. Then for a binary lottery, the system of equations we have to solve is 

\[
\begin{cases}
d+p(b-d) = y,\\
p(1-p)(b-d) = x.
\end{cases}
\]

But this is clearly solvable in $p,b,d$. Pick for example $p=\frac{1}{2}$, which leads to $b=y+2x, d=y-2x$. 
\end{proof}

This Lemma shows that the designer in the multi-agent model can use loss aversion in the money dimension to give additional incentives for truth-telling once the agents are locked-in after a positive participation decision. Indeed, for a fixed agent $i$ and type $\theta_i$ in the auction model the distributions $H$ correspond to the distributions on $\R$ generated as an image distribution of $t_i(\theta_i,\theta_{-i})$ under the product measure $F^{N-1}(\theta_{-i})$. Note that because of our assumptions this measure has full-support on $[\underline{\theta},\bar\theta]^{N-1}$. It follows that one can induce any Borel measure $H$ which is absolutely continuous w.r.t. the Lebesgue measure by appropriate choice of the (deterministic) payment functions $t_i$. Note that the expected news utility term $\omega$ here corresponds to $M$ in the screening model which was exogenously given there. In the multi-agent model $\omega$ becomes endogenous. 

\paragraph{An abstract Incentive Compatibility characterization}
\begin{proposition}\label{thm:icabstract}
A direct mechanism $\mathcal{M}$ as in \eqref{eq:directM} is incentive compatible if and only if 
the corresponding perceived valuations $\mathcal{W}_i$ are non-decreasing.
\\\indent If this is the case, we have the following Mirrlees representation for the interim utility in equilibrium $\mathcal{V}_i(\theta_i)= \mathcal{V}_i(\theta_i|\theta_i)$
\begin{equation}\label{eq:mirrleesnewsdelay}
\mathcal{V}_i(\theta_i) = \mathcal{V}_i(\underline{\theta}_i|\underline{\theta}_i) +  \int_{\underline{\theta}_i}^{\theta_i} \mathcal{W}_i(s)ds.
\end{equation}
It also holds 
\begin{equation}\label{eq:upsilondelay}
\Upsilon_i(\theta_i) = \mathcal{W}_i(\theta_i)\theta_i - \mathcal{V}_i(\underline{\theta}) - \int_{\underline{\theta}_i}^{\theta_i} \mathcal{W}_i(s)ds.
\end{equation}

\end{proposition}

The proof is a trivial adaptation of the proof of the Mirrlees representation in the classical quasilinear utility model. See \cite{boergers} for the classical proof. 

\section{Proofs for Section \ref{sec:screening}}

We start with the analysis for timeline B as that for timeline A is very similar, yet simpler. Analysis for timeline C is more involved due to the wedge between individual rationality and incentive compatibility.

\subsubsection{Analysis for timeline B}
Let $U(s) = (1+s)[\Gamma^B(s)v(q(s))-t(s)]=(1+s)W(s)$ be the utility of buyer of type $s$ in an incentive compatible mechanism. Then individual rationality is tantamount to 
\begin{equation*}
U(s)\geq 0 \text{ equivalent to } W(s)\geq 0,\quad \bar\lambda\geq s\geq 1.
\end{equation*}
Just as in the classical setting (see chapter 2 of \cite{boergers}) one can show that $W$ is differentiable a.e. with
\begin{align*}
W'(s) &= \frac{U'(s)}{1+s}-\frac{U(s)}{(1+s)^2}\\
& = \frac{m-m-M}{1+s}v(q(s))-\frac{t(s)}{1+s}-\frac{-t(s)+\Gamma^B(s)v(q(s))}{1+s}\\&= -\frac{2(m+M)v(q(s))}{(1+s)^2}.
\end{align*}
Here we have used the envelope theorem for the optimization problem of the agent. 
For any incentive compatible mechanism it follows with the Mirrlees representation that
\[
W(s)-W(\bar\lambda) = \int_{\bar\lambda}^sW'(t)dt = 2(m+M)\int^{\bar\lambda}_s\frac{v(q(t))}{(1+t)^2}dt.
\]
We can therefore write 
\begin{equation*}
t(s) = \Gamma^B(s)v(q(s))-2(m+M)\int^{\bar\lambda}_s\frac{v(q(t))}{(1+t)^2}dt-W(\bar\lambda),\quad \lambda_0\geq s\geq 1.
\end{equation*}
The fact that $W$ is decreasing for any incentive compatible mechanism implies that individual rationality is fulfilled if and only if $W(\bar\lambda)\geq 0$. An optimal mechanism will have $W(\bar\lambda) = 0$. 

In all, the profit function is 
\begin{equation*}
\Pi = \int_{1}^{\bar\lambda}\left[ \Gamma^B(s)v(q(s))-2(m+M)\int^{\bar\lambda}_s\frac{v(q(t))}{(1+t)^2}dt-cq(s)\right]G(ds).
\end{equation*}
The problem of the monopolist consists of maximizing this expression w.r.t.\footnote{`w.r.t.' means `with respect to'.} $q$ non-increasing. Doing the usual Fubini transformation for the double integral one gets
\begin{equation*}
\Pi = \int_{1}^{\bar\lambda}\left[ \Psi^B(s)v(q(s))-cq(s)\right]G(ds),
\end{equation*}
with 
\[
\Psi^B(s) = \Gamma^B(s)-2(m+M)\frac{G(s)}{(1+s)^2g(s)}.
\]
$\Psi^B(s)$ is the virtual valuation in this model. If it is non-increasing, then one can maximize the integrand point-wise to get the optimal solution. In that case one can calculate the FOC
\[
\mu\left[\Psi(s)v'(q(s)) - c\right] = 0,
\]
where $\mu$ is a Kuhn Tucker parameter, which is zero, as long as $\Psi^B(s)\leq 0$ and otherwise it is positive. 

The corresponding assumption of regularity from the classical setting (see \cite{myerson} or \cite{boergers}) is the following. 

\paragraph{Regularity for timeline B}\footnote{This just corresponds to the derivative of $\psi^B$ being non-positive.}

\begin{equation}\label{eq:reffi}
\frac{4m+7M}{2(m+M)} \ge  \frac{G(\lambda)}{g(\lambda)}\left[\frac{2}{1+\lambda}+\frac{g'(\lambda)}{g(\lambda)}\right],\quad\lambda \in [1,\bar\lambda]. 
\end{equation}

If $\Psi$ is not non-increasing, then one can use ironing techniques from \cite{toikka} to solve for bunching. Conditions in his paper are fulfilled since the model here is \emph{separable} according to his terminology (see section 3 of his paper).

\subsubsection{Analysis for timeline A}

The calculations for timeline A are virtually the same as for timeline B, except that the terms involving $M$ are missing because there is no realization effect in timeline A. The formal analysis is word-for-word the same as in timeline B except that now we have to set $M=0$ everywhere in the calculations.  

In particular, the regularity assumption for timeline A doesn't depend on the features of the distribution of intrinsic values $F$.
\paragraph{Regularity for timeline A}

\begin{equation*}
2 \ge  \frac{G(\lambda)}{g(\lambda)}\left[\frac{2}{1+\lambda}+\frac{g'(\lambda)}{g(\lambda)}\right],\quad\lambda \in [1,\bar\lambda]. 
\end{equation*}
Finding the optimal mechanism proceeds the same way as for timeline B. 

\subsubsection{Analysis for timeline C}



\begin{proof}
[Proof of Proposition \ref{thm:negativef}, Part 1.]

\emph{Step 0.} Assume in this preliminary step that the growth condition on $v$ implies that $$|x-y|\le const|v(x)^p-v(y)^p|,$$ where $const$ in the following will be a non-specified positive constant number which may change from line to line and always so that the respective inequality holds. It follows then that 
\[
|v^{-1}(x)-v^{-1}(y)|\le c|x^p-y^p|.
\]
This in turn implies that the functional 
\begin{equation}\label{eq:helpi}
L^p([1,\bar\lambda])\ni f\ra \int_{1}^{\bar\lambda} v^{-1}(f(s))ds 
\end{equation}
is continuous in the $L^p$-norm. The simple argument for this uses the fact that the function $[0,\infty)\ni x\ra x^{\frac{1}{p}}$ is H\"older continuous with exponent $\frac{1}{p}$. 

\paragraph{Proof of $|x-y|\le const|v(x)^p-v(y)^p|$} This follows simply by Taylor's expansion of order one applied to $x\ra v(x)^p$ and the growth condition on $v$. 
\vspace{2mm}\\
\emph{Step 1.} Since the virtual type $\Gamma^C(\lambda)=m+(1-\lambda)M$ is decreasing in $\lambda$, just as for the timeline B, it follows

\[
\text{Incentive compatibility is equivalent to } q(\cdot) \text{ non-decreasing}.
\]
By Mirrlees representation of the optimal payments for this situation, payments $t(\lambda)$ which ensure incentive compatibility are given up to a constant by the incentive compatibility  requirement of the decision-play self who has already decided to participate in the mechanism. We denote by $U(\lambda)$ the play-decision utility in equilibrium of an incentive compatible mechanism (this is determined in the second period). It is given by 

\[
U(\lambda) = \Gamma^C(\lambda)v(q(\lambda))-t(\lambda).
\]
The envelope theorem from the maximization problem of the agent gives 
\[
U'(\lambda) = -Mv(q(\lambda)). 
\]
We can use this to write 

\[
U(\lambda) = \Gamma^C(\lambda)v(q(\lambda))-t(\lambda)=U(1) - M\int_{1}^{\lambda}v(q(s))ds. 
\]
We solve for $t(\lambda)$ to write for a constant $f=-U(1)\in\R$. 
 
\begin{equation}\label{eq:transferC-screening}
t(\lambda) = f+ \Gamma^C(\lambda)v(q(\lambda))+M\int_1^{\lambda}v(q(s))ds. 
\end{equation}

\paragraph{Fact:} $t(\cdot)$ is weakly decreasing in any incentive compatible mechanism. 

This can be easily established through taking (one-sided) derivatives. The calculation uses Assumption (S) extensively.

The \textbf{Fact} has an important implication for the optimal mechanism whenever it exists: \emph{$t(\lambda)\ge 0$ always in any optimal mechanism for timeline C.} Proof of this implication is through contradiction. Suppose, for the sake of contradiction that there is an incentive compatible and individually rational mechanism which is profit-optimal for timeline C and so that $t(\lambda_0)<0$ for some $\lambda_0$. The Fact implies that $t(\lambda)<0$ for every $\lambda\ge \lambda_0$. But then the monopolist can switch to $(q(\lambda),t(\lambda)) = (0,0)$ for all $\lambda\ge\lambda_0$ (i.e. exclude all types with $\lambda\ge\lambda_0$). This preserves incentive compatibility and also individual rationality as the former is equivalent to $q$ non-increasing (which remains intact) and the latter depends only on the own type.\footnote{Thus, types outside $[\lambda_0,\bar\lambda]$ retain individual rationality under the changed mechanism and types in $[\lambda_0,\bar\lambda]$ are excluded from the mechanism and thus get their outside option of zero utility.} Thus there is a threshold type $\hat\lambda\in[1,\bar\lambda]$ so that the monopolist sets $(q(\lambda),t(\lambda)) = (0,0)$ whenever $\lambda>\hat\lambda$. 
\vspace{2mm}\\
\emph{Step 2.} The proof of this step doesn't need the growth condition on $v$. Assume in this step that the threshold type $\hat\lambda$ is already given. As one can see easily, the arguments of this step hold the same for every particular value 
$\hat\lambda\in (1,\bar\lambda]$.

Given the \emph{$t(\lambda)\ge 0$} - property for any optimal mechanism, at the participation stage (before the first delay), if the agent of type $\lambda$ accepts the mechanism the agent experiences the following utility (see also discussion in main text):

\[
V(\lambda) = \left[(2m+(1-\lambda)M\right]v(q(\lambda))-(1+\lambda)t(\lambda),
\]
or using the definitions from timeline B
\begin{equation}\label{eq:transfff}
V(\lambda) =  (1+\lambda)\Gamma^B(\lambda)v(q(\lambda))-(1+\lambda)t(\lambda).
\end{equation}
Replacing \eqref{eq:transferC-screening} into \eqref{eq:transfff} we get

\begin{equation*}
V(\lambda)=(1+\lambda)\left[\Gamma^B(\lambda)-\Gamma^C(\lambda)\right]v(q(\lambda))-(1+\lambda)f-(1+\lambda)M\int_1^{\lambda}v(q(s))ds. 
\end{equation*}
It follows that the individual rationality requirement can be written as

\[
\left[\Gamma^B(\lambda)-\Gamma^C(\lambda)\right]v(q(\lambda))-f-M\int_1^{\lambda}v(q(s))ds\geq 0,\quad \lambda\in [1,\hat\lambda]. 
\]



The objective function of the designer becomes

\begin{equation}\label{eq:R-Cscreening}
\int_1^{\hat\lambda}t(s)-cq(s)dG(s) = f\cdot G(\hat\lambda)+\int_1^{\hat\lambda} \left[\Gamma^C(\lambda)v(q(\lambda))+M\int_1^{\lambda}v(q(s))ds\right]-cq(s)dG(\lambda).
\end{equation}
After the usual application of Fubini's Theorem this expression turns into 

\begin{equation*}
f\cdot G(\hat\lambda) + \int_1^{\hat\lambda} \left\lbrace\left[\Gamma^C(\lambda)+ M\frac{G(\hat\lambda)-G(\lambda)}{g(\lambda)}\right]v(q(\lambda))-cq(\lambda)\right\rbrace dG(\lambda).  
\end{equation*}

Thus the problem of the designer for timeline C can be rewritten as follows

\begin{align}
\begin{split}\label{eq:R-Cscreening21}
\max_{f\in\R,s\mapsto q(s)}& \quad f\cdot G(\hat\lambda) + \int_1^{\hat\lambda} \left\lbrace\left[\Gamma^C(\lambda)+ M\frac{G(\hat\lambda)-G(\lambda)}{g(\lambda)}\right]v(q(\lambda))-cq(\lambda)\right\rbrace dG(\lambda)\\&
\text{s.t. }\\& (1)\text{ } q(s)\text{ is non-increasing,}\\& (2)\text{ }
\left[\Gamma^B(\lambda)-\Gamma^C(\lambda)\right]v(q(\lambda))-f-M\int_1^{\lambda}v(q(s))ds\geq 0,\quad \lambda\in [1,\hat\lambda].
\end{split}
\end{align}

But note that by taking $f$ in \eqref{eq:R-Cscreening} into the integral\footnote{$f\cdot G(\hat\lambda) = \int_1^{\hat\lambda}fdG(\lambda)$.} and replacing the constraint (2) we get an upper bound for the profit. Namely, the maximum of the following integral

\begin{equation}\label{eq:upperbd}
\int_1^{\hat\lambda} \Gamma^B(\lambda)v(q(\lambda))-cq(\lambda) dG(\lambda),
\end{equation}
with respect to all non-increasing $q$. This program is easily solvable through point-wise maximization given our constraints.\footnote{Recall that $g$ is bounded away from zero on $[1,\bar\lambda]$.} In in all we have an upper bound for the program in \eqref{eq:R-Cscreening2}. The point-wise maximum of \eqref{eq:upperbd} has the following property:

\begin{itemize}
\item[] (!) $\infty>q(\lambda)>0$ whenever $\Gamma^B(
\lambda)>0$ and $q(\lambda) = 0$ otherwise. This means the problem in \eqref{eq:upperbd} has a finite value. 

In particular, it is not optimal in \eqref{eq:upperbd} to set $q$ so that $\limsup_{\lambda\ra 1^+}q(\lambda) = +\infty$, i.e. $q$ is not locally unbounded near $1$. That would only be profitable if the optimized value in \eqref{eq:R-Cscreening2} were infinity, which contradicts (!). This implies that $v(q(\lambda)),\lambda\in [1,\hat\lambda]$ remains bounded for any potential solution of \eqref{eq:R-Cscreening2}. 
\end{itemize}

Moreover, it is easy to see from \eqref{eq:R-Cscreening2} and (!) that 
\begin{itemize}
\item[] (!!) $f$ is bounded and finite in any potential optimum. 
\end{itemize}
 
Introduce now for a given nondecreasing $q(\cdot)$ 

\begin{equation}\label{eq:helpu}
u(\lambda) = f+M\int_1^{\lambda}v(q(s))ds.
\end{equation}

This results in $u$ being continuous, non-decreasing and concave, but with free endpoint $u(1) = f$. Moreover, for every $u$ non-decreasing and concave there exists a non-increasing $q(\cdot)$ and a $f\in\R$ such that $u$ can be expressed as in \eqref{eq:helpu}. We make use of this equivalence considerably in what follows. 

Inserting \eqref{eq:helpu} in \eqref{eq:R-Cscreening} we get that each feasible allocation $q(\cdot)$ for the program \eqref{eq:R-Cscreening2} results in a feasible `allocation' $u(\cdot)$ for the following maximization program. 

\begin{align}
\begin{split}\label{eq:R-Cscreening3}
\max_{s\mapsto u(s)}& \quad \int_1^{\hat\lambda} \left\lbrace\Gamma^C(\lambda)\frac{u'(\lambda)}{M}+u(\lambda)-cv^{-1}(u'(\lambda))\right\rbrace dG(\lambda)\\&
\text{s.t. }\\& (1)\text{ } u(s)\text{ is continuous, non-decreasing, concave}\\& (2)\text{ }
\left[\Gamma^B(\lambda)-\Gamma^C(\lambda)\right]\frac{u'(\lambda)}{M}-u(\lambda)\geq 0,\quad a.e.\quad\lambda\in [1,\hat\lambda].
\end{split}
\end{align}

The other direction is also trivially true: whenever we have a solution $u$ to \eqref{eq:R-Cscreening3} we get a solution to \eqref{eq:R-Cscreening21} by the relations $q(s) = v^{-1}\left(\frac{u'(s)}{M}\right)$, $t(s) = u(s) + \Gamma^C(s)\frac{u'(s)}{M}$ for $s\le \hat\lambda$ and $q(s) = 0, t(s) = 0$ otherwise. Note that it is w.l.o.g. to ask for constraint $(2)$ to hold a.e. This is because $G$ is absolutely continuous w.r.t. the Lebesgue measure on $[1,\bar\lambda]$. 

Because of (!) and (!!) we can assume that $u'(1)$ and $u(1)$ are bounded for all feasible candidates of \eqref{eq:R-Cscreening3}. In combination with constraint $(1)$ in \eqref{eq:R-Cscreening3} (continuity and concavity) this implies that we can restrict the class of $u$ in the maximization in \eqref{eq:R-Cscreening3} by requiring additionally the following condition.\footnote{$u'$ is bounded from below by monotonicity. If it were optimal to have $\limsup_{s\ra 1}u'(s) = +\infty$ we would get a contradiction to (!) since $u'(s) = Mv(q(s))$.}

(3) functions $u$ are bounded in the Sobolev p-norm: $u\mapsto \left(\int_1^{\hat\lambda}|u(\lambda)|^pd\lambda+\int_1^{\hat\lambda}|u'(\lambda)|^pd\lambda\right)^{\frac{1}{p}}$ by a constant $c>0$. Here recall that $p> 1$. The constant $c$ is chosen to be independent of $\hat\lambda$ given its boundedness and the upper and lower bounds for the optimal profit of timeline C. Moreover they are also bounded in the essential supremum-norm by a uniform bound holding for all $u$ under consideration. 
\vspace{2mm}\\
Denote by $\mathcal{C}(\hat\lambda)$ the class of functions satisfying $(1), (2), (3)$ where $(1)$ and $(2)$ are only required to hold almost everywhere.

Note then the easy-to-show but important facts: $\mathcal{C}(\hat\lambda)$ is convex and bounded in the Sobolev-p norm. This follows from the linearity of the two constraints and of the Sobolev norm, as well as the requirement (3). 

$\mathcal{C}(\hat\lambda)$ is closed in the Sobolev-p norm.\footnote{\label{fn:ftnote}Note that there is a natural continuous and onto embedding $\iota(\hat\lambda, \hat\lambda'):\mathcal{C}(\hat\lambda')\ra \mathcal{C}(\hat\lambda)$ whenever $\hat\lambda\le \hat\lambda'$. It is given by restricting the definition domain of the functions.
} To show this assume that we have a sequence $u_n, n\ge 1$ which satisfy $(1), (2), (3)$ and which converge in Sobolev-p norm to $u$. Then that $(2), (3)$ hold true for $u$ is easy to check: convergence in $L^p$ implies convergence of a suitable subsequence a.e.; this implies that $u$ is non-decreasing and concave $[1,\hat\lambda]$ which means that it is also continuous in $(1,\hat\lambda)$; $(1),(2)$ are then automatically fulfilled; $(3)$ is fulfilled due to the definition of convergence in the Sobolev-p norm. 


Since $\mathcal{C}(\hat\lambda)$ is closed and bounded it is sequentially compact in the weak topology induced by Sobolev p-norm. This is because the space of Sobolev functions equipped with the Sobolev p-norm is reflexive and one can apply the Banach-Alaoglu Theorem which implies sequential compactness w.r.t. weak topology for all bounded, closed sets.\footnote{See e.g. section 3.15 of \cite{rudin}.}


The functional being maximized has the form 
\[
J(u) = \int_1^{\hat\lambda}L(\lambda,u(\lambda),u'(\lambda))d\lambda
\]
with $L(\lambda, u,w) = g(\lambda)\left(\Gamma^C(\lambda)\frac{w}{M}+u-cv^{-1}(w)\right)$ for $(\lambda,u,w)$ coming from a bounded set of $\R^3$ (determined by the above discussion) and otherwise flattens to zero continuously for all $(\lambda,u,w)$ outside this bounded set. 
$L$ is then a continuous function and satisfies 
\[
L(\lambda,u,w)\leq a(\lambda)w +b(\lambda) + c|u|,
\]
with some bounded functions $a,b>0$. 
Finally, $w\mapsto L(\lambda,u,w)$ is strictly concave since $v^{-1}$ is strictly convex. By Theorem 3.23, pg. 96 in \cite{dacorogna} it follows that $J$ is sequential upper semi-continuous w.r.t. the weak topology. 


Because of upper-semicontinuity of $J$ and compactness of the set $\mathcal{C}(\hat\lambda)$ (all in the same topology), we know that maximizing $J$ over $\mathcal{C}(\hat\lambda)$ has a solution. The solution is a Sobolev function, i.e. it can be chosen to be continuous everywhere and differentiable almost everywhere.

\emph{Step 2.} We now show how to calculate the optimal $\hat\lambda\in [1,\bar\lambda].$ For this, it suffices to show that the value function of the maximization problem \eqref{eq:R-Cscreening3} is continuous w.r.t. $\hat\lambda\in [1,\bar\lambda]$. 

Note that $[1,\bar\lambda]\ni\hat\lambda\ra\mathcal{C}(\hat\lambda)$ is a continuous, compact-valued correspondence in $\hat\lambda$. Here we are considering all $\mathcal{C}(\hat\lambda)$ embedded in the `largest' space $\mathcal{C}(\bar\lambda)$, all equipped with the Sobolev-p norm. Upper hemicontinuity follows by checking directly its sequential characterization. Lower hemicontinuity follows just as easily because of the continuous and onto embedding $\iota(\hat\lambda, \hat\lambda'):\mathcal{C}(\hat\lambda')\ra \mathcal{C}(\hat\lambda)$ introduced in footnote \ref{fn:ftnote}. Finally, note that the objective functional in \eqref{eq:R-Cscreening3} is jointly continuous in $(\hat\lambda,u)$ where $[1,\bar\lambda]\times (Sobolev-p)([1,\bar\lambda])$ is equipped with the product topology of the euclidean one and the one of the Sobolev-p norm. This follows because of the H\"older inequality combined with the continuity of the map in \eqref{eq:helpi}. 

\end{proof}

\begin{proof}[Proof of Proposition \ref{thm:negativef}, part 2)]
This follows from the fact that $\Gamma^B(\lambda)-\Gamma^C(\lambda) = \frac{1-\lambda}{1+\lambda}(m-\lambda M)$. $\Gamma^B(\lambda)-\Gamma^C(\lambda)$ is zero for $\lambda = 1$ and negative for all small $\lambda>1$. It follows from the second constraint in \eqref{eq:R-Cscreening2} that $f\leq 0$ in any optimal menu of contracts. If it were true that $q(1) = 0$ then it would follow that $q\equiv 0$, i.e. the profit would be zero. Thus, whenever the profit is not zero it means that for some range near $\lambda=1$ the monopolist offers positive amounts of the good. It follows that $f<0$ must hold. 
\end{proof}

\begin{proof}[Proof of Theorem \ref{thm:AbeatsBscreening}] 
That $A$ is weakly better than $B$ follows directly from the fact that the virtual valuation $\Psi^A$ in timeline A is pointwise higher than the virtual valuation $\Psi^B$ in timeline B and that otherwise the set of incentive compatible and individually rational mechanisms is the same for both timelines. 

That B is weakly better than C follows from Proposition \ref{thm:negativef}. Namely, we show there during the existence proof for timeline C that:

- the profit in timeline C is always bounded by the profit in timeline B, leaving out the individual rationality requirement in timeline B (proof of part 1 of Proposition \ref{thm:negativef}). The \emph{optimal} IC mechanism in timeline B is automatically IR. 

- Individual rationality in timeline C weakly costs something positive to the monopolist whenever profits in C are strictly positive (part 2) of Proposition \ref{thm:negativef}).

\end{proof}

\section{Proofs for Section \ref{sec:multiagent}}

\begin{proof}[Proof of Proposition \ref{thm:IC}.] Straightforward algebra calculations show the formulas for \eqref{eq:mirrleesgeneral} hold for both the case when $v_i(\emptyset)$ is the minimum and the case when it is the maximum of $v_i$ with the definitions made in the main text. We have just split payments $t_i(\theta_i,\cdot)$ into the range where it is strictly positive and the range where it is negative. The positive range is multiplied by $\lambda_i^m$. A similar split is done for the good dimension, but now it is the negative range which is multiplied by $\lambda_i^g$. One uses Lemma \ref{thm:expectedNU} and algebra.
The same steps can be made for the terms $\omega_i$. 

\end{proof}

\begin{proof}[Proof of Proposition \ref{thm:icgeneral}]
Here we use Proposition \ref{thm:icabstract} in a first step which ensures the existence of the perceived payments. 

The proof is finished if we show the existence of the payment schedules $t_i$ for fixed $\Upsilon_i$-s delivered by the proof in the classical setting. Some $t_i$-s that do the trick are the following: set first $t_i(\theta_i,\theta_{-i}) = t_i(\theta_i,\theta_{-i}'),\forall \theta_i,\theta_{-i},\theta_{-i}'$. This implies $T_i(\theta) = t_i(\theta,\theta_{-i}), \forall \theta,\theta_{-i}$ and $T_i^+(\theta) = t_i(\theta,\theta_-i)\textbf{1}_{\{t_i(\theta,\theta_-i)\geq 0\}} =T_i(\theta)\textbf{1}_{\{T_i(\theta)\geq 0\}}$. Then, $\omega_i(\theta)= 0$. If $\Upsilon_i(\theta)=0$ then set $T_i(\theta)=0$, otherwise if $\Upsilon_i(\theta)>0$ set $T_i(\theta) = \frac{\Upsilon_i(\theta)}{1+\lambda_i^m\mu_i^m}$, while if $\Upsilon_i(\theta)<0$ then set $T_i(\theta) = \frac{\Upsilon_i(\theta)}{1+\mu_i^m}$. Note that $T_i(\theta)$ are increasing in $\Upsilon(\theta_i)$, albeit with a discontinuity at zero. This finishes the proof, since the Mirrlees representation also follows by the same arguments as in the classical proof with quasilinear utilities. In general, if one wants a mechanism with non-trivial $\omega_i$-s one uses Lemma \ref{thm:lmhelpmixed}. In timeline B, where the $T_i^+$ terms matter, we show in the online appendix that for optimal mechanisms it holds $\omega_i\equiv 0$. It follows that optimal mechanisms in all of our applications below will have $t_i\ge 0$, so that $T_i^+ =T_i$.\footnote{For timelines A and C this is clear from Proposition \ref{thm:IC}.}  
\end{proof}

\begin{proof}[Proof of Proposition \ref{thm:prop1public}] In the following we suppress the timeline superscript whenever the argument is valid for all timelines or it is clear from the proof context to which timeline the statements correspond. 

1)-2) The case of timeline A is clear from the discussion in text. We show the result for timeline $B$. The proof for timeline C is similar. 

First note that $q$ being symmetric we can write $Q_i(\theta)=Q(\theta)$ and thus $\mathcal{W}_i =\mathcal{W}$ for each $i$. Moreover, $q$ being monotone increasing in each $\theta_i$, the same follows for $Q$. Indeed, we have $Q(\theta) = 1- F^{*(N-1)}(N\tilde{c}(N)-\theta)$. Taking $\mathcal{W}$ as a function of Q, i.e. $\mathcal{W}(Q) = (1+\mu^g)Q-\Lambda^gQ(1-Q)$ we have 
\begin{equation*}
\frac{d\mathcal{W}}{dQ} = 1+\mu^g -\Lambda^g+ 2\Lambda^gQ.
\end{equation*}
This derivative is always nonnegative for $\Lambda^g\leq 1+\mu^g$, which establishes 1).

For $\Lambda^g>1+\mu^g$, substituting the formula for $Q$ into the derivative above, we see that $\mathcal{W}$ is nondecreasing if and only if
\begin{equation}\label{eq:pubhelp1}
\frac{1}{2}(\frac{1+\mu^g+\Lambda^g}{\Lambda^g})\geq F^{*(N-1)}(N\tilde{c}(N)-\theta) \quad \text{for all }\theta \in [\underline{\theta},\overline{\theta}].
\end{equation}
The left hand side of \eqref{eq:pubhelp1} is smaller than $1$, if the premise of 2) is fulfilled. Also the condition is more likely to be violated as $\tilde{c}(N)$ approaches $\overline{\theta}$ and as $F$ becomes more and more concentrated in lower $\theta$-s. Note also that the upper bound on the left of \eqref{eq:pubhelp1} is smaller the higher $\Lambda^g$ is. This establishes 2).

3) Note first the bound $F^{*(N-1)}(N\tilde{c}(N)-\theta) \leq F^{*(N-1)}(N\tilde{c}(N)-\underline{\theta})$, which is uniform in $\theta$. We can write 
\begin{equation*}
F^{*(N-1)}(N\tilde{c}(N)-\underline{\theta}) = \mathbb{P}\large(\frac{\sum_{i=1}^{N-1}\theta_i}{N-1}\leq \frac{N\tilde{c}(N)-\underline{\theta}}{N-1}\large)
\end{equation*}
Now we only need to note that due to the \emph{Strong Law of Large Numbers} \footnote{See for example chapter 2 of \cite{durrett}.} we know that $\frac{\sum_{i=1}^{N-1}\theta_i}{N-1}$ converges almost surely to $\mathbb{E}[F]$. Meanwhile, $\frac{N\tilde{c}(N)}{N-1}$ has possible limit points all strictly smaller than $\mathbb{E}[F]$ due to assumption. Thus we have that as $N\rightarrow \infty$ the left hand side of \eqref{eq:pubhelp1} converges to zero uniformly in $\theta$. 
\end{proof}

\begin{proof}
[Proof of Theorem \ref{thm:optimaltimelineresults}, 1)]
We know the rule for the optimal direct mechanism reads \emph{`give the good to the bidder $i$ with the highest $\theta_i$, as long as $\theta_i>\theta^{*}$'}. Individual rationality is realized in both timelines by setting the utility of the lowest types in equilibrium equal to zero. Recalling that in both cases the revenue in the optimal mechanism is given by 

\begin{equation*}
\sum_{i=1}^N \int_{\underline{\theta}}^{\bar{\theta}}\mathcal{W}^{A/B}_i(\theta_i)\gamma(\theta_i)dF(\theta_i), 
\end{equation*}

one sees that the fact $\mathcal{W}^{A}_i(\theta_i)\geq \mathcal{W}^{B}_i(\theta_i),$ together with the fact that in both cases the optimal allocation rules prescribe $\mathcal{W}^{A/B}_i(\theta_i) = 0,$ whenever $\theta_i<\theta^*$, the result follows.
\end{proof}

\paragraph{Analysis for timeline C in the case of symmetric auctions.}

We assume throughout i.i.d. intrinsic types and that $\Lambda^g\leq 1$. 

Given this, the symmetric equilibrium of any incentive compatible auction has a threshold type: sell to the highest type above a threshold. This follows from Proposition \ref{thm:icgeneral} and the fact that it is without loss of generality to look at symmetric allocation rules. The reason for the latter is the same as in classical setting (see footnote 11 in \cite{maskinriley}). It follows for incentive compatible allocation rules that the expected probability of getting the good when own type is $s$ has the following form.

\[
Q(s) = F(s)^{n-1}\textbf{1}_{\{s\geq\hat\theta\}}. 
\]

Given this, we can then build the perceived valuation 

\[
\mathcal{W}(s) = Q(s)(1-\Lambda^g(1-Q(s))),
\]
as well as define the auxiliary functions 

\[
h(s) = \mathcal{W}(s)s - \int_{\underline{\theta}}^s\mathcal{W}(t)dt.
\]

Define furthermore 

\[
g(s) = \mathcal{W}(s)s + \mu^gQ(s)s. 
\]
Note that $g$ and $h$ depend on the threshold $\hat\theta$. They are both zero below $\hat\theta$. Moreover, both functions are weakly increasing due to IC and the assumption that $\Lambda^g\le 1$. Denote in the following by $c$ the decision utility of the agent with the lowest type $\underline{\theta}$.

\textbf{Claim.} An incentive compatible and individually rational mechanism in the timeline C is a tuple $(\hat\theta,\omega,c)$ s.t. 

\[
(IC)\quad c+\Lambda^m\omega(s) +T(s) = h(s),\quad s\ge \hat\theta
\]
\[
(IR)\quad c\geq s_{m,\hat\theta}(s)-\frac{\lambda^m\mu^m}{1+\lambda^m\mu^m}\Lambda^m\omega(s), \quad s\ge \hat\theta
\]
where $s_{m,\hat\theta}(s) = h(s) - \frac{1}{1+\lambda^m\mu^m}g(s)$. 

\begin{proof}[Proof of Claim]
These follow directly by using Propositions \ref{thm:icgeneral} and \ref{thm:IRgeneral}. (IC) is immediate whereas the individual rationality requirement can be written first as 

\begin{equation}\label{eq:ircontrolp}
g(\theta)\geq (1+\lambda^m\mu^m)T(\theta) + \Lambda^m\omega(\theta),\quad s\ge \hat\theta
\end{equation}
which is then easily manipulated into (IR). 
\end{proof}
Note that (IC) together with Assumption (A2) implies that $\Lambda^m\omega(s) +T(s)$ is increasing in $s$. Thus, types with $s<\hat\theta$ don't get served at all and pay nothing whereas types $s\ge\hat\theta$ may receive a (net) subsidy of $c$. 

$s_{m,\hat\theta}$ is a piece-wise smooth function with at most one discontinuity of uniformly bounded size across all incentive compatible mechanisms.\footnote{It is increasing in the money friction $\lambda^m\mu^m$. Its difference to $h$ disappears uniformly as $\lambda^m\mu^m\ra\infty$.}

Due to symmetry and the above we can rewrite the maximization problem of the designer as 
\[
\max_{c\in\R,(T,\omega):\Theta_i\ra\R\times\R_{+}} \int_{\underline{\theta}}^{\bar\theta} T(\theta)dF(\theta),\quad s.t.\quad  (IC)\text{ and }(IR)\quad \forall \theta.
\]



The objective function of the problem for timeline C is 

\[
\int_{\hat\theta}^{\bar\theta}h(t)dF(t) - \int_{\underline{\theta}}^{\bar\theta} (c+\Lambda^m\omega(s))dF(s). 
\]

We can split the maximization problem in two parts. Once a threshold $\hat\theta$ has been chosen, the rest of the mechanism is found by solving 

\[
\min_{c,\omega(\cdot)\geq 0} \int_{\hat\theta}^{\bar\theta} (c+\Lambda^m\omega(s))dF(s),
\]

under the constraint

\[
(IR)\quad c\geq s_{m,\hat\theta}(s)-\frac{\lambda^m\mu^m}{1+\lambda^m\mu^m}\Lambda^m\omega(s),\quad s\ge \hat\theta.
\]
 

Given any $c$, one can see that the optimal $\omega$ has to satisfy

\[
\Lambda^m\omega(s) = \frac{1+\lambda^m\mu^m}{\lambda^m\mu^m}\max\{s_{m,\hat\theta}(s)-c,0\},\quad s\ge \hat\theta
\]

Note that payments are degenerate for all types (all-pay) if and only if $c\geq \max_{s}s_{m,\hat\theta}(s)$. Otherwise, the optimal mechanism is not all-pay and some of the types will not be fully insured in the money dimension.  

Overall, given a threshold type $\hat\theta$ the rest of the mechanism is determined by solving 
\[
\min_{c\in \R}\frac{1+\lambda^m\mu^m}{\lambda^m\mu^m}\int_{\hat\theta}^{\bar\theta}\max\{s_{m,\hat\theta}(s)-\frac{1}{1+\lambda^m\mu^m}c,\frac{\lambda^m\mu^m}{1+\lambda^m\mu^m}c\}dF(s)- F(\hat\theta)\max\{-\frac{1}{\lambda^m\mu^m}c,c\}
\]
Note that the presence of the last term in the minimization problem will usually make for a non-smooth solution as one varies the parameter $\lambda^m\mu^m$. Nevertheless, the value function of this minimization problem is a continuous function of $\hat\theta$.\footnote{The conditions for Berge's maximum theorem are given because $c$ can be taken to be bounded without loss of generality, whenever $\lambda^m\mu^m$ comes from a bounded interval.} Denote the value function of this problem by $H(\hat\theta)$. The optimal threshold then solves 

\[
\max_{\hat\theta\in[\underline{\theta},\bar\theta]}\int_{\hat\theta}^{\bar\theta}h(t)dF(t) - H(\hat\theta).  
\]
The maximand is a continuous function being maximized over a compact interval. Therefore there always exists a solution so that the discussion above implies that there always exists an optimal auction for timeline C.

\paragraph{Optimal Timelines}

Given that timeline B is never optimal the comparison is between timelines A and C. 

Recall that the revenue from one auction participant in timeline A is 

\[
\frac{1+\mu^g}{1+\lambda^m\mu^m}\int_{\underline{\theta}}^{\bar\theta} \left(Q(s)s-\int_0^sQ(t)dt\right)dF(t)=\frac{1+\mu^g}{1+\lambda^m\mu^m}\int_{\theta^*}^1Q(s)\gamma(s)ds 
\]
Here, $\theta^*$ satisfies $\gamma(\theta^*) = 0$. The recipe for the optimal revenue in the case of timeline C is given in the preceding paragraph. 

The rest of the work for the examples is numerical analysis using the software package R.

\end{appendices}

\end{document}